\documentclass{article}
\usepackage[margin=1.5cm, includefoot, footskip=30pt]{geometry}
\usepackage{amsmath}
\usepackage{amsfonts}
\usepackage{amssymb}
\usepackage{authblk}
\usepackage[nocompress]{cite}
\usepackage{graphicx}
\usepackage{hyperref}
\usepackage{subcaption}
\usepackage{tikz}
\usetikzlibrary{calc}
\usetikzlibrary{shapes.geometric}
\usetikzlibrary{arrows.meta, positioning}
\usepackage{latexsym}

\newcommand{\gt}{\ensuremath >}
\usepackage{float}
\usepackage{framed}
\usepackage{mdframed}
\usetikzlibrary{positioning}
\usetikzlibrary{matrix}
\usepackage{amsthm}

\newtheorem{theorem}{Theorem}
\newtheorem{lemma}{Lemma}
\newtheorem{definition}{Definition}

\title{The Cooperation Ceiling: Extrinsic Population Dynamics and the Intrinsic Escape}
\author[1,*]{Harry Foster}
\author[1]{Vince Knight}
\author[2]{Sebastian Krapohl}
\affil[1]{School of Mathematics, Cardiff University}
\affil[2]{Faculty of Social and Behavioural Sciences, University of Amsterdam}
\affil[*]{Corresponding author: \texttt{fosterh3@cardiff.ac.uk}}
\date{\today}

\begin{document}
\maketitle

\begin{abstract}
    Evolutionary game theory provides a framework by which to study the
    emergence of cooperation in a population of self-interested actors. In such
    a framework, players' decisions on whether or not to cooperate evolve
    according to decision rules called population dynamics. However, often games
    are studied under the assumption that all individuals play under the same
    conditions, and many common choices of update rule are not well suited for a
    heterogeneous population. In this paper, we categorise and compare four
    different population dynamics in such a population as ``extrinsic'', where
    players learn by looking outward at the payoffs of other
    players, and ``intrinsic'', where players look inwardly at their own
    attributes or potential payoffs. We show that extrinsic population dynamics
    admit a ceiling on the rate of cooperation which can be exceeded by
    intrinsic population dynamics, and demonstrate this using the public goods
    game with heterogeneous contributions.
\end{abstract}

\section{Introduction}

The conditions under which cooperation can emerge among self-interested actors
is a central question across biology, economics, and the social sciences. Many
of the situations in which this question arises are collective action problems,
in which individuals share an interest in a common good yet each would prefer
that others bear the cost of providing it. The structure of such a dilemma maps
onto a \emph{public goods game}~\cite{ARCHETTI20129, hauert2002volunteering,
hauser2019social, HauertDynamics2004, 2bbcfd92-fcd5-38d1-b7a0-bc6a2aadbc64}:
each actor decides whether to incur a private cost to contribute to a shared
resource whose total value is multiplied by some factor \(r\) and distributed
equally amongst all players. This creates the temptation to free-ride, and so
purely selfish reasoning leads to a tragedy of the
commons~\cite{hardin1968tragedy, Weitz043299}, even when collective cooperation
would benefit everyone. Both theory~\cite{santos2011risk} and controlled
experiments~\cite{Milinski2011,milinski2008collective} show that framing
cooperation as a collective-risk problem, where failure to meet a cooperation
threshold incurs a shared loss, can substantially raise cooperation rates.

Evolutionary game theory offers a framework for studying how cooperation can
emerge among self-interested actors without assuming rational
foresight~\cite{hofbauer_sigmund_1998, 568a8a9c-d8ae-3d24-9fc2-3d33340230e7}. In
this approach, strategies that perform well spread through a population via one
of several update rules, while poorly performing strategies decline. The
conditions under which cooperative strategies can invade and stabilise have been
extensively studied in two-player settings such as the iterated prisoner's
dilemma~\cite{b5d8eed6-6ac7-35c7-953f-9af57df3f45e} and the snowdrift
game~\cite{hauert_spatial_2004}. The public goods game generalises these pairwise
dilemmas to \(N\) players~\cite{hauert2002volunteering}, making it well suited
to studying multi-lateral cooperation. In finite populations, evolutionary
outcomes are stochastic: a rare mutant strategy spreads or vanishes with a
probability that depends on the population size, payoff structure, and update
rule~\cite{Nowak2004Emergence, traulsen2006stochastic}. 

The conditions under which one strategy invades a population have been studied
in depth for particular evolutionary processes. The \(\frac{1}{N}\) fixation
probability of a single mutant under neutral
drift~\cite{568a8a9c-d8ae-3d24-9fc2-3d33340230e7} serves as a benchmark for when
selection favours a strategy, and results such as the \(\frac{1}{3}\)
rule~\cite{Nowak2004Emergence,OHTSUKI2007289} hold even when a population
updates with both the Moran process and Fermi imitation
dynamics~\cite{LIU2017336}. For symmetric \(2\times2\) games, the condition
\(a(N-2) + bN > cN + d(N-2)\)~\cite{Kandori, tarnita_strategy_2009} on the
payoff entries leads to the abundance of a given strategy in a wide range of
evolutionary processes and for arbitrary mutation and selection
intensity~\cite{ANTAL2009340}. These results,
however, concern a symmetric game played pairwise between identical members of a
population. We instead consider all \(N\)-player games in which a defector
outperforms a cooperator in a given state, amongst a fully heterogeneous
population.

A feature of real-world collective action that standard models typically set
aside is that actors are not identical. Players differ in wealth, productive
capacity, and the costs they face, and this inequality can reshape the dynamics
of cooperation~\cite{hauser2019social}. One commonly studied heterogeneous
population is one in a spatial setting, where players are placed on a network
and participate simultaneously in multiple games with their
neighbours~\cite{nowak1992evolutionary,
perc2013evolutionary,szabo2007evolutionary}. In~\cite{santos2008social} it is
shown that social diversity in endowments promotes the emergence of cooperation
in the public goods game: heterogeneous populations outperform their homogeneous
counterparts over a wide range of parameters. This echoes earlier work showing
that scale-free network topology, a form of structural heterogeneity, provides a
unifying framework for the emergence of cooperation~\cite{santos2005scale}. When
contributions depend on dynamically accumulated reputation~\cite{MA2021111353},
high-reputation defectors receive reduced transfers, weakening defection's
advantage. When heterogeneity arises from differences in network degree so that
players allocate contributions according to the difference in the number of
connections~\cite{LEI20104708}, the payoff structure either penalises
medium-degree defectors, or creates small clusters of strong cooperators. The
characteristic spatial mechanism is the formation of cooperating clusters that
are more profitable than their defecting neighbours; provided \(r\) is
sufficiently large, these clusters resist invasion and eventually spread.
Heterogeneity does not, however, unconditionally favour cooperation.
Flores~et~al.~\cite{PhysRevE.108.024111} show that when the contribution level
is itself a strategic variable rather than a fixed attribute, low-value
contributions act as a form of defection against the most collectively
beneficial strategies, so that the emergence of cooperation does not guarantee
the emergence of high-value cooperation.

How the choice of update rule interacts with population heterogeneity is also
poorly understood. The Moran process~\cite{knight2018evolution,Moran_1958},
Fermi imitation dynamics~\cite{PhysRevE.58.69}, aspiration
dynamics~\cite{du2014aspiration}, and introspection dynamics~\cite{Couto2022,
Couto2023} all model strategy revision in distinct ways, and each may respond
differently to player-specific payoffs. In particular, \emph{intrinsic}
dynamics, in which a player evaluates a prospective strategy change based on
their payoff in comparison with their own counterfactual payoff or aspiration
level, are better suited to heterogeneous populations than \emph{extrinsic}
dynamics, in which a player may copy a high-fitness neighbour without
considering whether the same strategy would serve them as well.

We make the following contributions. First, we build on an established general
framework for evolutionary dynamics on fully heterogeneous, well-mixed
populations of \(N\) ordered individuals, in which each player may have a
distinct payoff function and contribution level. Second, we formalise the notion
of a purely extrinsic population dynamic, a definition encompassing both the
Moran process~\cite{knight2018evolution,Moran_1958} and Fermi imitation
dynamics~\cite{PhysRevE.58.69}, and show that such dynamics are inherently
constrained by an upper bound on the probability of cooperation. Third, we
consider two purely intrinsic population dynamics, aspiration
dynamics~\cite{du2014aspiration} and introspection dynamics~\cite{Couto2022,
Couto2023}, showing that these processes can overcome the ceiling of purely
extrinsic dynamics. For each population dynamic, we use the public goods
game~\cite{ARCHETTI20129, hauert2002volunteering, HauertDynamics2004,
hauser2019social, 2bbcfd92-fcd5-38d1-b7a0-bc6a2aadbc64} to demonstrate our
results.

The remainder of the paper is organised as follows.
Section~\ref{sec:the_general_heterogeneous_population_dynamic_model} introduces
the population model and the public goods game.
Section~\ref{sec:extrinsic_population_dynamics} introduces a formal definition
of a \emph{purely extrinsic population dynamic}, gives formulations of the Moran
process and Fermi imitation dynamics in terms of our population model, proves
that these satisfy the definition of a purely extrinsic population dynamic, and
proves the ceiling on the probability of cooperation in such dynamics. We then
demonstrate this using a heterogeneous public goods game.
Section~\ref{sec:intrinsic_population_dynamics} gives formulations of two
intrinsic population dynamics and shows that they can exceed the ceiling on the
probability of cooperation, demonstrated using their results in the public goods
game. Finally, Section~\ref{sec:robustness_to_population_size} confirms through
direct simulation of the Markov chain that these results persist for large
populations.

\section{The general heterogeneous population dynamic model}
\label{sec:the_general_heterogeneous_population_dynamic_model}

As in~\cite{Couto2022, Couto2023}, we define a general framework for
evolutionary dynamics on a fully heterogeneous, well-mixed population.

\subsection{The state space}

We model a population consisting of \(N\) \textbf{ordered} individuals. Each
individual selects an action from a common action set \(\mathcal{A} = (\mathcal{A}_1, \mathcal{A}_2,
\ldots,\mathcal{A}_k)\) of size \(k\). A \emph{state} of the population is an ordered
\(N\)-tuple \(\mathbf{a} = (a_1, a_2, \ldots, a_N) \in \mathcal{A}^N\), and the
\emph{state space} \(S\) is the set of all such tuples. A payoff function \(\pi
\colon S \to \mathbb{R}^N\) assigns to each state a vector of individual
payoffs; we write \(\pi_i(\mathbf{a})\) for the payoff of individual \(i\) in
state \(\mathbf{a}\).

The state space of an evolutionary game can be described as the graph whose
nodes are the states \(\mathbf{a} \in S\), with an edge between two vertices
\(\mathbf{a}, \mathbf{b}\) if and only if they differ in exactly one position.
This defines an evolutionary process in which one player changes action type at
each time step. Each state also has a self-loop, as at a given time step, a
player may accept the same strategy they currently possess, or reject the
strategy which they have considered. We call the unique coordinate at which
\(\mathbf{a}\) and \(\mathbf{b}\) differ the \emph{index of difference}, denoted
\(I(\mathbf{a}, \mathbf{b})\). The set of states reachable from \(\mathbf{a}\)
in one step is denoted \(\mathrm{Neb}(\mathbf{a}) = \{\mathbf{b} \in S :
h(\mathbf{a}, \mathbf{b}) = 1\}\). This notation follows~\cite{Couto2023}.

In the case of a game with two actions, the state space of the Markov chain is
analogous to an \(N\)-dimensional hypercube. This is the graph whose node set
consists of the \(2^N\) \(N\)-dimensional boolean vectors, and which has edges
between vertices if and only if they differ in exactly one position
~\cite{HARARY1988277}. The state space of an evolutionary game with two actions
can be modelled in this way by denoting the strategies \(0\) and \(1\).
Throughout, the two-action games we study are social dilemmas, so we take these
two actions to be cooperate (\(C\)) and defect (\(D\)) and write the state space
as \(\{C, D\}^N\). The graphical representation of a birth-death process differs
from a hypercube in exactly two ways. First, a population may remain the same at
a given step, and thus each vertex contains a self-loop. Second, the probability
of transitioning from a state \(\mathbf{a}\) to a state \(\mathbf{b}\) is not
necessarily the same as the probability of transitioning from \(\mathbf{b}\) to
\(\mathbf{a}\). Thus, instead of one edge between two states, we have two
directed edges, \(\mathbf{a} \to \mathbf{b}\) and \(\mathbf{b} \to \mathbf{a}\).
Figure~\ref{fig:4-cube}(a) and~\ref{fig:4-cube}(b) illustrate the population
model and its state space.

\subsection{Population dynamics and mutation}

The population evolves as a Markov
chain~\cite{f66ea2a5-f1c3-3855-8dc0-e3047d83298e} on \(S\). We therefore obtain
a transition matrix \(T\), where \(T_{\mathbf{ab}}\) denotes the probability of
transitioning from a state \(\mathbf{a}\) to a state \(\mathbf{b}\), as shown in
figure~\ref{fig:4-cube}(c). They take the general form:

\begin{equation}
    T_{\mathbf{ab}} =
    \begin{cases}
        P(\mathbf{a} \to \mathbf{b})
        & \text{if } \mathbf{b} \in \text{Neb}(\mathbf{a}), \\
        0
        & \text{if } \mathbf{b} \notin \text{Neb}(\mathbf{a}) \text{ and } \mathbf{a} \neq \mathbf{b}, \\
        1 - \sum_{\mathbf{b} \in S \setminus \{\mathbf{a}\}} T_{\mathbf{ab}}
        & \text{if } \mathbf{a} = \mathbf{b}.
    \end{cases}
\end{equation}

where \(P(\mathbf{a} \to \mathbf{b})\) is the probability that player
\(I(\mathbf{a,b})\) changes from its action type in \(\mathbf{a}\) to its action
type in \(\mathbf{b}\). The process which determines this probability is called
a \emph{population dynamic}. We shall discuss four population dynamics in this
paper, and classify them as shown in Figure~\ref{fig:diagrams_of_dynamics}.
Purely extrinsic population dynamics only consider the success of a strategy in
the current population when updating a player's action type. This models a
behaviour where players look outwardly to attempt to improve their payoff.
Purely intrinsic population dynamics model a player who looks inwardly to
improve their payoff, comparing their current payoff to some value which does
not depend on the payoffs of an action type in the current population. The
dynamics discussed in this paper are shown in figure~\ref{fig:diagrams_of_dynamics}

\begin{figure}[!htbp]
    \centering
    \begin{mdframed}
    \includegraphics[width=\linewidth]{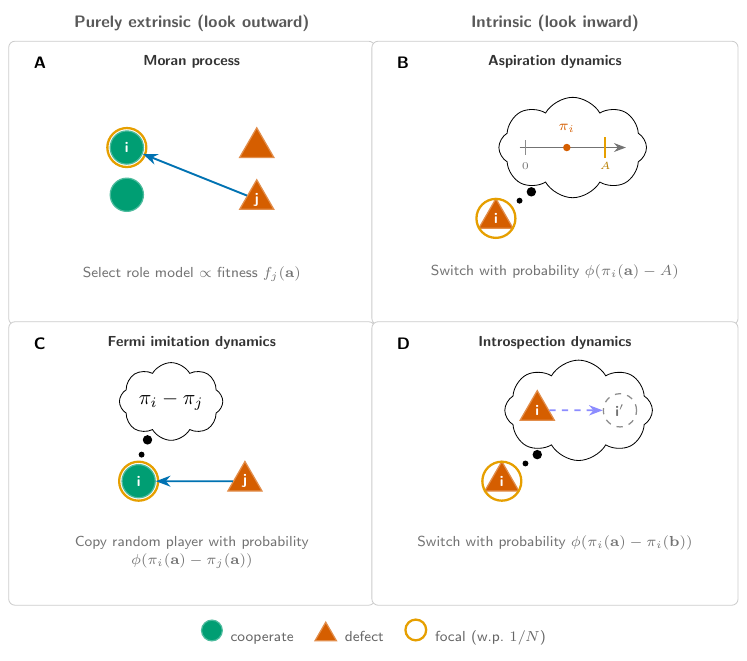}
    \end{mdframed}
    \caption{\textbf{The four population dynamics.} \textbf{A,} the Moran
    process and \textbf{C,} Fermi imitation dynamics are purely extrinsic.
    \textbf{B,} aspiration dynamics and \textbf{D,} introspection dynamics are
    purely intrinsic. For introspection dynamics the candidate action is chosen
    with probability \(\frac{1}{k-1}\), and for aspiration dynamics there is
    only one other action to accept.}
    \label{fig:diagrams_of_dynamics}
\end{figure}

Two of the dynamics discussed in this paper, the Moran process and Fermi
imitation dynamics, correspond to absorbing Markov chains, where the measurable
outcome is the probability of absorption into a given absorbing state. The
remaining two, introspection and aspiration dynamics, are ergodic chains, for
which we can compute the long-run abundance of a given strategy in the
population. To enable meaningful comparison across all models, we introduce
mutation. Mutation ensures that all processes become ergodic, while also
introducing a controlled amount of noise: at any moment individuals may switch
to a different strategy~\cite{FUDENBERG2006352}. Consequently, we no longer
obtain absorption matrices for any model with mutation; instead, we derive a
unique steady state that is independent of the initial state.

We assume a
mutation parameter \(\mu \in [0, 1]_{\mathbb{R}}^{N \times k}\), where
\(\mu_{ij}\) denotes the probability that individual \(i\) adopts strategy \(j\)
via mutation. The stationary abundance of a strategy \(C\), denoted \(p_C\), is
well defined only with mutation for the Moran process and Fermi imitation
dynamics, so we report these for \(\mu > 0\); the introspection and aspiration
dynamics are ergodic for every \(\mu \geq 0\). The figures throughout use the
mutation rates \(\mu \in \{0.001, 0.005, 0.05, 0.1\}\).

Given a process defined by transition matrix \(T\) and two states \(\mathbf{a}\)
and \(\mathbf{b}\), the probability of the original process \emph{with mutation}
transitioning from \(\mathbf{a}\) to \(\mathbf{b}\) is given by:

\[
    P^{(\mu)}(\mathbf{a} \to \mathbf{b})
    = P(\text{pick individual } I(\mathbf{a}, \mathbf{b})) \cdot \left(
        P(\mathbf{a}_{I(\mathbf{a,b})} \text{ transitions} \to \mathbf{b}_{I(\mathbf{a,b})}) \cdot
        P(\text{no mutation})
        + P(\mathbf{a}_{I(\mathbf{a}, \mathbf{b})}\text{ mutates} \to
        \mathbf{b}_{I(\mathbf{a}, \mathbf{b})})
    \right)
    \]

For a focal individual \(i\), the probability that no
mutation occurs is \(1 - \sum_{j=1}^k \mu_{ij}\). Let \(T\) be the transition
matrix for a population dynamic and \(T_{\mathbf{ab}}\) be the transition
probability from a state \(\mathbf{a}\) to a state \(\mathbf{b}\). Then the
transition matrix \(T^{(\mu)}\) for the process with mutation is given by:

\begin{equation}
    T^{(\mu)}_{\mathbf{ab}} =
    \begin{cases}
        T_{\mathbf{ab}} \left(1 - \sum_{j=1}^k \mu_{I(\mathbf{a}, \mathbf{b}), j}\right)
        + \frac{\mu_{I(\mathbf{a}, \mathbf{b}), \mathbf{b}_{I(\mathbf{a}, \mathbf{b})}}}{N}
        & \text{if } \mathbf{b} \in \text{Neb}(\mathbf{a}), \\
        0
        & \text{if } \mathbf{b} \notin \text{Neb}(\mathbf{a}) \text{ and } \mathbf{a} \neq \mathbf{b}, \\
        1 - \sum_{\mathbf{b} \in S \setminus \{\mathbf{a}\}} T_{\mathbf{ab}}
        & \text{if } \mathbf{a} = \mathbf{b}.
    \end{cases}
\end{equation}

This behaviour is shown in figure~\ref{fig:4-cube}(d). Note that if
\(\sum_{j=1}^k \mu_{ij} = 1\), mutation fully determines the behaviour of
individual \(i\). All processes in this paper are considered
under the effect of action-invariant mutation, where \(\mu_{iw} = \mu_{il}\) for
all action types \(w,l\). This ensures that mutation does not cause an inherent
bias towards a given action type.

\begin{figure}[!htbp]
    \begin{mdframed}
    \centering
    \includegraphics[width=\linewidth]{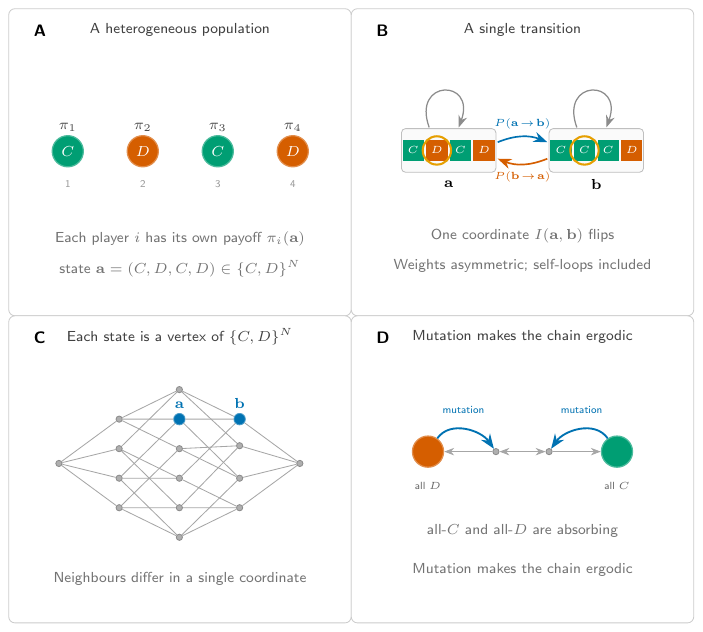}
    \end{mdframed}
    \caption{\textbf{The heterogeneous population as a Markov chain.}
    \textbf{A,} A heterogeneous population with action set \(\mathcal{A} = \{C, D\}\), in
    which each player has its own payoff function. \textbf{B,} A single
    transition between neighbouring states. \textbf{C,} Each state is a vertex
    of \(S = A^N\), so the state space is an \(N\)-dimensional hypercube.
    \textbf{D,} Mutation connects the chain and makes it ergodic.}
    \label{fig:4-cube}
\end{figure}

\section{Extrinsic population dynamics}
\label{sec:extrinsic_population_dynamics}
We begin by defining a \emph{purely extrinsic population dynamic}. Intuitively,
a purely extrinsic process revises strategies through payoff comparison between
members of a population.

\begin{definition}[Purely extrinsic population dynamic]
    Let \(T\) be the transition matrix of a population dynamic on a state space
    \(S\). We say that \(T\) defines a \emph{purely extrinsic} process if for
    every \(\mathbf{a,b} \in S\) with \(\mathbf{b} \in \text{Neb}(\mathbf{a})\),
    \(T_{\mathbf{ab}}\) satisfies:

    \begin{enumerate}
        \item \(T_{\mathbf{ab}}\) is a function of payoffs evaluated in
        \(\mathbf{a}\) and in no other state.
        \item Let \(\mathcal{R}\) denote the set of players \(j\) such that
        \(j \neq I(\mathbf{a,b})\) and \(\mathbf{a}_j = \mathbf{b}_{I(\mathbf{a,b})}\).
        Then \(T_{\mathbf{ab}}\) is non-decreasing in the payoffs
        \(\pi_j(\mathbf{a})\) for all \(j \in \mathcal{R}\) (and strictly
        increasing in at least one such \(j\) if \(\mathcal{R} \neq \emptyset\)),
        and non-increasing in the payoffs \(\pi_j(\mathbf{a})\) for all
        \(j \notin \mathcal{R}\).
        \item \(T_{\mathbf{ab}}\) contains no term dependent on \(a_{I(\mathbf{a,b})}\) or
        \(b_{I(\mathbf{a,b})}\), except in the payoff function
        \(\pi_{I(\mathbf{a,b})}(\mathbf{a})\), mutation parameters \(\mu_{ij}\),
        and in the definition of the set \(\mathcal{R}\).
    \end{enumerate}
    \label{def:purely_extrinsic}
\end{definition}

Condition (1) gives the property of purely extrinsic population dynamics that
looks outwardly, copying other individuals in the state. Condition (2)
captures the basic monotonicity: a higher payoff for a role model makes that
role model more likely to be copied, and a higher payoff for a non-role-model
pulls the other way. Condition (3) rules out an inherent bias towards either
action. Without (3), we must consider population dynamics which explicitly
favour a certain action type, for example:

\[
\bar{T}_{\mathbf{ab}} =
\begin{cases}
\delta_{I(\mathbf{a}, \mathbf{b})} \cdot T_{\mathbf{ab}} & \text{if } \mathbf{b} \neq \mathbf{a}, \\
1 - \sum_{\mathbf{c} \neq \mathbf{a}} T_{\mathbf{ac}} & \text{if } \mathbf{b} = \mathbf{a}.
\end{cases}
\]

where \(T_{\mathbf{ab}}\) is the transition matrix of some purely extrinsic
dynamic, and \(\delta_{I(\mathbf{a,b})}\) takes the value 1 if
\(b_{I(\mathbf{a,b})} = C\) and \(\frac{1}{2}\) if not. Without (3), \(\bar{T}\)
would define a purely extrinsic population dynamic, but would be tailored to
favour a particular action type. We exclude cases such as this, as they do not
conform to the framework in which fitness is the single measure of the success,
and thus the propagation, of action types in the population. While
process-specific parameters may affect the transition probability
\(T_{\mathbf{ab}}\), they should not inherently favour any specific action type.

We now show two classical examples of purely extrinsic population dynamics: the
Moran process, in which a player is chosen to be duplicated proportional to
their payoff, and Fermi imitation dynamics, which directly compares the payoffs
of pairs of players.

\subsection{The Moran process}

The Moran process~\cite{Moran_1958} is defined by the following algorithm:
\begin{enumerate}
    \item A player is chosen to be copied with a probability proportional to 
    their fitness in the population
    \item A player is chosen with probability \(\frac{1}{N}\) to change their
        strategy
\end{enumerate}

This is shown in Figure~\ref{fig:diagrams_of_dynamics}A, and the transition matrix is given by:

\begin{equation}
    T_{\mathbf{ab}} =
    \begin{cases}
        \dfrac{\sum_{j:\, a_j = b_{I(\mathbf{a}, \mathbf{b})}} f_j(\mathbf{a})}{N\sum_j f_j(\mathbf{a})}\left(1 - \sum_{j=1}^k \mu_{I(\mathbf{a}, \mathbf{b}), j}\right)+ \frac{\mu_{I(\mathbf{a}, \mathbf{b}), \mathbf{b}_{I(\mathbf{a}, \mathbf{b})}}}{N} & \text{if } \mathbf{b} \in \mathrm{Neb}(\mathbf{a}),\\
        0 & \text{if } \mathbf{b} \notin \mathrm{Neb}(\mathbf{a}) \text{ and } \mathbf{a} \neq \mathbf{b},\\
        1 - \sum_{\mathbf{b} \in S \setminus \{\mathbf{a}\}} T_{\mathbf{ab}} & \text{if } \mathbf{a} = \mathbf{b}.
    \end{cases}
\label{eqn:Moran_Probability_Function}
\end{equation}

Where \(f_i(\mathbf{a}) = 1 - \varepsilon + \varepsilon\pi_i(\mathbf{a})\) is
the fitness of a player \(i\) in state \(\mathbf{a}\) scaled with selection
intensity \(\varepsilon\in [0,1]\). The value of \(\varepsilon\) is chosen to ensure a
positive fitness for all players in order that the probability
\(T_{\mathbf{ab}}\) is always well defined. This parameter controls the
rationality with which players choose which strategy to adopt in each time step:
a larger \(\varepsilon\) indicates that a player will choose a higher fitness
strategy with a greater probability.

The Moran process has often been employed as a model of biological
evolution~\cite{568a8a9c-d8ae-3d24-9fc2-3d33340230e7}, where a player's ``type''
represents a species or allele within a population. In each time step, rather
than a given player choosing a new action type, a randomly chosen individual in
the population would die and be removed from the population. At the same time,
an individual would reproduce with a probability proportional to their fitness,
introducing an exact copy of themselves to the population to replace the removed
individual. As a result of the model's original purpose, the Moran process
contains no decision step where a player chooses whether or not to accept a
considered strategy. Instead, at each time step, a player \emph{must} accept
another player's action, and we remain in the same state when a player accepts
the same action as they currently play.

One implication of this is that in a heterogeneous population, worsening moves
may be taken with a high probability, as an action that provides a high payoff
for one player in a population may not provide the same return for another. In
nature, an offspring is unable to choose its parents, and thus cannot choose
which attributes it is born with, and so the Moran process does not model the
decision which players have regarding whether or not to play a certain strategy.
When modelling the adoption of behaviours, these assumptions may lead to players
making decisions which are harmful for their own payoff.

We now show that the Moran process satisfies the conditions of
Definition~\ref{def:purely_extrinsic}.

\begin{theorem}[The Moran process is a purely extrinsic process]
    The Moran process, as defined by
    Equation~\eqref{eqn:Moran_Probability_Function}, is a purely extrinsic
    population dynamic.
\end{theorem}

\begin{proof}
By definition the Moran process satisfies conditions (1) and (3). It remains to
be shown that \(T_{\mathbf{ab}}\) is non-decreasing in \(\pi_j(\mathbf{a})\) for \(j
\in \mathcal{R}\) (strictly increasing for at least one such \(j\)) and
non-increasing in \(\pi_j(\mathbf{a})\) for \(j \notin \mathcal{R}\). Recall
\(f_i(\mathbf{a}) = 1 - \varepsilon + \varepsilon \pi_i(\mathbf{a})\) for some selection intensity
\(\varepsilon > 0\). Then we have:

\[T_{\mathbf{ab}} = \frac{\sum_{j \in \mathcal{R}}
f_j(\mathbf{a})}{N\sum_{l=1}^N f_l(\mathbf{a})} \left(1 - \sum_{i=1}^k
\mu_{I(\mathbf{a,b}), i}\right) + \frac{\mu_{I(\mathbf{a,b}),
b_{I(\mathbf{a,b})}}}{N}\]
\[
\frac{dT_{\mathbf{ab}}}{d\pi_i(\mathbf{a})} = \frac{\varepsilon \sum_{x \notin \mathcal{R}}
f_x(\mathbf{a})}{\left(N\!\sum_{l=1}^N f_l(\mathbf{a})\right)^2}\left(1 - \sum_{i=1}^k \mu_{I(\mathbf{a,b}), i}\right)
\quad\text{for } i \in \mathcal{R}
\]

For \(i \in \mathcal{R}\), as \(\varepsilon\) is chosen such that \(f_i(\mathbf{a}) \gt
0\) for all \(i \in \mathbf{a}\), and \(\mu_{i,j}\) such that \(\sum_{i=1}^k
\mu_{I(\mathbf{a,b}), i} < 1\), this derivative is strictly positive. For \(j
\notin \mathcal{R}\), the analogous computation gives
\(\frac{dT_{\mathbf{ab}}}{d\pi_j(\mathbf{a})} = -\frac{\varepsilon \sum_{x \in
\mathcal{R}} f_x(\mathbf{a})}{(N\sum_{l=1}^N f_l(\mathbf{a}))^2}(1 -
\sum_{i=1}^k \mu_{I(\mathbf{a,b}), i}) \leq 0\), so condition (2) holds. 
\end{proof}

\subsection{Fermi imitation dynamics}

Fermi imitation dynamics~\cite{PhysRevE.58.69} is another commonly employed
method of modelling evolution in a population. It follows the process:

\begin{itemize}
  \item Two players, \(i\) and \(j\) are selected at random from the population.
  \item Player \(i\) copies the action type of the player \(j\) with a
  probability according to the Fermi function \( \phi(\pi_{i}(\mathbf{a}) -
  \pi_{j}(\mathbf{a})) = \frac{1}{1 + e^{\left({\beta(\pi_{i}(\mathbf{a}) -
  \pi_{j}(\mathbf{a}))}\right)}}\)
\end{itemize}

Where \(\beta\) is the \emph{choice} intensity of the system, representing how
sensitive our players' choice is to the difference between their current payoff
and the payoff which they are considering. This algorithm is also shown in
Figure~\ref{fig:diagrams_of_dynamics}C. This method has been commonly used in studies of
evolutionary games, for example~\cite{LEI20104708,MA2021111353,PhysRevE.58.69}.
We now define our transition matrix according to the Fermi function:

\begin{equation}
    T_{\mathbf{ab}} =
    \begin{cases}
        \dfrac{1}{N(N-1)}\!\sum_{j:\, a_j = b_{I(\mathbf{a}, \mathbf{b})}}\phi\!\left(\pi_{I(\mathbf{a}, \mathbf{b})}(\mathbf{a}) - \pi_{j}(\mathbf{a})\right)\left(1 - \sum_{j=1}^k \mu_{I(\mathbf{a}, \mathbf{b}), j}\right)+ \frac{\mu_{I(\mathbf{a}, \mathbf{b}), \mathbf{b}_{I(\mathbf{a}, \mathbf{b})}}}{N} & \text{if } \mathbf{b} \in \mathrm{Neb}(\mathbf{a}),\\
        0 & \text{if } \mathbf{b} \notin \mathrm{Neb}(\mathbf{a}) \text{ and } \mathbf{a} \neq \mathbf{b},\\
        1 - \sum_{\mathbf{b} \in S \setminus \{\mathbf{a}\}} T_{\mathbf{ab}} & \text{if } \mathbf{a} = \mathbf{b}.
    \end{cases}
\label{eqn:imitation_prob}
\end{equation}

Unlike a Moran Process, we can see that the players are not selected with a
probability proportional to their fitness. Instead, Fermi Imitation Dynamics selects
a pair of players and compares their payoffs directly, giving players the option
not to accept a strategy at a given time step. However, similarly to the Moran
process, in a heterogeneous population the method of adopting a new strategy
based entirely on the payoff of another individual may lead to worsening moves
in which players adopt strategies which currently admit a higher fitness, but
which will not improve their own fitness.

We also see that due to the formulation of the logit function \(\phi\), Fermi
imitation dynamics can become invariant under certain parameters. This occurs in
any game in which one action type's payoff in any given state can be expressed
as another action type's payoff up to the addition of some term \(\gamma\),
which depends only on a subset of the game's parameters. In that case,
\(\Delta(\pi) = \pm\gamma\), and is therefore invariant to any parameters of the
game of which \(\gamma\) is not a function.

We now show that Fermi imitation dynamics satisfies the definition of a purely
extrinsic population dynamic.

\begin{theorem}[Fermi imitation dynamics is a purely extrinsic population dynamic]
    Fermi imitation dynamics, as defined by
    Equation~\eqref{eqn:imitation_prob}, is a purely extrinsic population
    dynamic.
\end{theorem}

\begin{proof}
By definition, Fermi imitation dynamics satisfies conditions (1) and (3). It
remains to show that it satisfies condition (2).

\[ T_{\mathbf{ab}} = \frac{1}{N(N-1)}\sum_{j \in
\mathcal{R}}\phi(\pi_{I(\mathbf{a,b})}(\mathbf{a}) - \pi_j(\mathbf{a}))\left(1 -
\sum_{i=1}^k \mu_{I(\mathbf{a,b}), i}\right) + \frac{\mu_{I(\mathbf{a,b}),
b_{I(\mathbf{a,b})}}}{N} \]

As \(\phi\) is strictly decreasing, every term in the sum is strictly increasing
in \(\pi_j(\mathbf{a})\) for \(j \in \mathcal{R}\). For \(j \notin \mathcal{R}\) with
\(j \neq I(\mathbf{a,b})\), \(\pi_j(\mathbf{a})\) does not appear in the formula at
all. For \(j = I(\mathbf{a,b})\) (when \(a_{I(\mathbf{a,b})} \neq
b_{I(\mathbf{a,b})}\), so \(I \notin \mathcal{R}\)), each term is strictly
decreasing in \(\pi_{I(\mathbf{a,b})}(\mathbf{a})\) because \(\phi\) is. Thus, Fermi
imitation dynamics satisfies condition (2).
\end{proof}

\subsection{The cooperation ceiling}

We now consider purely extrinsic population dynamics in
games with two actions, \(C\) and \(D\). Figure~\ref{fig:ceiling_argument}
sketches the argument that follows: in a social dilemma a defector pointwise
dominates a cooperator, a relabelling of the two actions fixes a symmetric point
at neutral payoffs, and together these bound the cooperation probability at one
half. We first define a further classification for population dynamics:

\begin{figure}[!htbp]
    \centering
    \begin{mdframed}
    \includegraphics[width=\linewidth]{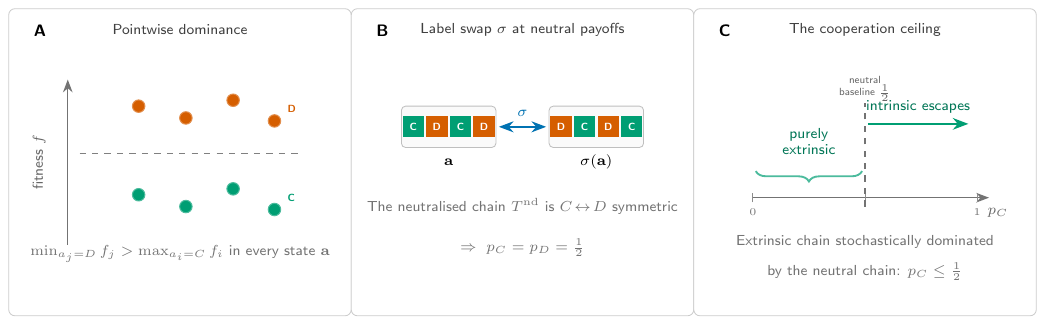}
    \end{mdframed}
    \caption{\textbf{Why purely extrinsic dynamics cannot exceed
    \(p_C = \frac{1}{2}\).} \textbf{A,} In a social dilemma a defector pointwise
    dominates a cooperator: in any state, switching to \(D\) does not lower the
    payoff. \textbf{B,} At neutral payoffs the dynamic is unchanged by the label
    swap \(\sigma\) that exchanges \(C\) and \(D\), fixing the symmetric point.
    \textbf{C,} Stochastic monotonicity then bounds the stationary cooperation
    probability at the neutral-drift baseline \(p_C = \frac{1}{2}\).}
    \label{fig:ceiling_argument}
\end{figure}

\begin{definition}[Neutrally monotone population dynamic]
    Let \(T\) be the transition matrix for a population dynamic on a game with
    two actions, \(\{C, D\}\), and equip the state space \(S =\{C, D\}^N\) with
    the componentwise partial ordering with \(D < C\). Let \(T^{\mathrm{nd}}\)
    denote \(T\) with all payoffs set to a common value (we denote by
    \(v_{\mathbf{a}}\) the common value in state \textbf{a}). We say that \(T\)
    defines a \emph{neutrally monotone} population dynamic if:
    
    \begin{enumerate}
    \item \(T_{\mathbf{ab}}^{nd}\) is independent of \(v_{\mathbf{a}}\).
    \item For every pair of states \(\mathbf{a} \leq \mathbf{a}'\), there exist random
        variables \(\mathbf{X}, \mathbf{X}'\) on a common probability space, with
        \(\mathbf{X} \sim T^{\mathrm{nd}}(\mathbf{a}, \cdot)\) and
        \(\mathbf{X}' \sim T^{\mathrm{nd}}(\mathbf{a}', \cdot)\), such that
        \(\mathbf{X} \leq \mathbf{X}'\) almost surely.

    \end{enumerate}
\end{definition}

The two conditions ask that the neutralised chain \(T^{\mathrm{nd}}\) is
independent of the common payoff value (condition (1)) and is stochastically
monotone (condition (2)). Neutral monotonicity thus captures configuration
neutrality: under the same neutral payoffs, starting from a more cooperative
profile cannot make the next state any less cooperative, and the chosen common
value does not impact this neutrality. In a process which is not neutrally
monotone, there could be a configuration-dependent
pull toward a given action type at neutral payoffs, which is an inherent bias
rather than truly neutral behaviour. It can be shown that this definition holds
for both the Moran process, and Fermi imitation dynamics:

\begin{theorem}
    The Moran process and Fermi imitation dynamics are neutrally monotone under action-invariant mutation.
\end{theorem}

\begin{proof}

    We begin with the Moran process. Set \(\pi_i(\mathbf{a}) = v\) for every
    player \(i\) gives \(f_i(\mathbf{a}) = 1 -\varepsilon + \varepsilon v\) for every \(i\),
    whence \(T_{\mathbf{ab}} = \frac{|\mathcal{R}|(1-\varepsilon+\varepsilon
    v)}{N\cdot N (1-\varepsilon+\varepsilon v)}(1 - \sum_i \mu_{I(\mathbf{a,b}),
    i}) + \frac{\mu_{I(\mathbf{a,b}), b_{I(\mathbf{a,b})}}}{N} =
    \frac{|\mathcal{R}|}{N^2}(1 - \sum_i \mu_{I(\mathbf{a,b}), i}) +
    \frac{\mu_{I(\mathbf{a,b}), b_{I(\mathbf{a,b})}}}{N}\), which is independent
    of \(v\). Therefore, condition (1) holds.

    Under \(T^{\mathrm{nd}}\), as defined above, the Moran process samples a focal player
    \(i\) uniformly, a role
    model \(k\) uniformly from the population (since all fitnesses are equal),
    and a mutation outcome \(m\) from the action-invariant mutation distribution;
    player \(i\) then adopts \(a_k\), or \(m\) if a mutation occurs. Couple the chains
    started at \(\mathbf{a} \leq \mathbf{a}'\) on the common probability space
    generated by \((i, k, m)\), applying the same draw to both. Coordinate \(i\) in
    the successor of \(\mathbf{a}\) is \(a_k\) (or \(m\)); coordinate \(i\) in the
    successor of \(\mathbf{a}'\) is \(a'_k\) (or the same \(m\)). Since
    \(\mathbf{a} \leq \mathbf{a}'\) gives \(a_k \leq a'_k\), and every other
    coordinate \(j\) is left unchanged with \(a_j \leq a'_j\), the two successors
    \(\mathbf{X}\) and \(\mathbf{X}'\) satisfy \(\mathbf{X} \leq \mathbf{X}'\)
    pointwise on this probability space, so condition (2) holds.

    For Fermi imitation dynamics, if \(\pi_i(\mathbf{a}) = v\) for every player
    \(i\), then every payoff difference appearing in the Fermi formula is zero, so
    \(\phi(\pi_I(\mathbf{a}) - \pi_j(\mathbf{a})) = \phi(0) = \frac{1}{2}\) in
    each summand, and \(T_{\mathbf{ab}} = \frac{|\mathcal{R}|}{2N(N-1)}(1 -
    \sum_i \mu_{I(\mathbf{a,b}), i}) + \frac{\mu_{I(\mathbf{a,b}),
    b_{I(\mathbf{a,b})}}}{N}\), independent of \(v\). Therefore, condition (1)
    holds.

    Under \(T^{\mathrm{nd}}\), the dynamic samples a focal player \(i\) uniformly, a
    role model \(k\) uniformly from the remaining \(N - 1\) players, an independent
    Bernoulli outcome with success probability \(\phi(0) = \frac{1}{2}\), and a
    mutation outcome \(m\) from the action-invariant mutation distribution; player
    \(i\) adopts \(a_k\) on a success (or \(m\) if a mutation occurs) and otherwise
    keeps \(a_i\). Couple the chains started at \(\mathbf{a} \leq \mathbf{a}'\) on
    the common probability space generated by \((i, k, \text{Bernoulli}, m)\),
    applying the same draw to both. On a mutation step, both successors set
    coordinate \(i\) to the same \(m\). On a successful imitation step, they set
    coordinate \(i\) to \(a_k\) and \(a'_k\) respectively, with \(a_k \leq a'_k\). On a
    no-update step, coordinate \(i\) retains \(a_i \leq a'_i\). Every other
    coordinate is unchanged, so the two successors \(\mathbf{X}\) and
    \(\mathbf{X}'\) satisfy \(\mathbf{X} \leq \mathbf{X}'\) pointwise on this
    probability space, and condition (2) holds.
\end{proof}

We will now show that there is a theoretic limit on the abundance of cooperators
in any neutrally monotone purely extrinsic population dynamic. Before stating the
theorem we record a comparison result for Markov chains on partially ordered
spaces. The result is classical and we claim no novelty; we include a short
proof for the reader's convenience, since its use in the proof of the theorem
below may not be familiar to the reader.

\begin{lemma}\label{lem:stationary-dominance}
    Let \(P\) and \(Q\) be irreducible, aperiodic transition matrices on a
    finite partially ordered set \((S, \leq)\), with stationary distributions
    \(\rho^P\) and \(\rho^Q\). Suppose that
    \begin{enumerate}
        \item[(i)] \(P(\mathbf{a}, \cdot) \preceq_{st} Q(\mathbf{a}, \cdot)\) for
        every \(\mathbf{a} \in S\); and
        \item[(ii)] \(Q\) is stochastically monotone: \(\mathbf{a} \leq
        \mathbf{a}'\) implies \(Q(\mathbf{a}, \cdot) \preceq_{st} Q(\mathbf{a}',
        \cdot)\).
    \end{enumerate}
    Then \(\rho^P \preceq_{st} \rho^Q\).
\end{lemma}

\begin{proof}
    Kamae, Krengel, and O'Brien~\cite{kamae1977stochastic} prove that
    whenever two probability measures \(\nu_1, \nu_2\) on a partially ordered
    Polish space satisfy \(\nu_1 \preceq_{st} \nu_2\), there exist random
    variables \(X_1 \sim \nu_1\) and \(X_2 \sim \nu_2\) defined on a common
    probability space with \(X_1 \leq X_2\) almost surely (a \emph{monotone
    coupling}).

    Using this, we construct a coupling of the two chains by induction on \(t\).
    Start with any \(X^P_0, X^Q_0 \in S\) satisfying \(X^P_0 \leq X^Q_0\). Given
    \(X^P_t \leq X^Q_t\), hypothesis (ii) gives \(Q(X^P_t, \cdot) \preceq_{st}
    Q(X^Q_t, \cdot)\), which combined with (i) at \(\mathbf{a} = X^P_t\) yields
    \(P(X^P_t, \cdot) \preceq_{st} Q(X^Q_t, \cdot)\). Applying
    Kamae--Krengel--O'Brien to this pair of measures produces \((X^P_{t+1},
    X^Q_{t+1})\) with the required marginals and \(X^P_{t+1} \leq X^Q_{t+1}\)
    almost surely. By induction, \(X^P_t \leq X^Q_t\) for every \(t\). Since both
    chains are ergodic, their marginal laws converge to \(\rho^P\) and \(\rho^Q\), and
    passing to the limit gives \(\rho^P \preceq_{st} \rho^Q\).
\end{proof}

We now state the ceiling on the average abundance of cooperators in purely
extrinsic population dynamics. Our result is comparable to that
of~\cite{ANTAL2009340}, who studied the symmetric \(2\times2\) game with payoff
matrix
\begin{equation}
\begin{array}{c|cc}
    & A & B \\ \hline
A & a & b \\
B & c & d
\end{array}
\label{eq:abcd}
\end{equation}
and showed that, for a wide range of evolutionary processes with mutation, the
condition \(a(N-2) + bN > cN + d(N-2)\) makes strategy \(A\) more abundant than
\(B\) at any mutation rate. This builds upon the work of~\cite{Kandori,
Nowak2004Emergence}, generalising results to arbitrary mutation rates and
selection intensity. The impact of mutation on the abundance of a given strategy
in such a game is discussed by Fudenberg~et~al~\cite{FUDENBERG2006352}. We now
show that for neutrally monotone purely extrinsic population dynamics, pointwise
domination of defection over cooperation is sufficient for defection to dominate
cooperation in a heterogeneous population.

\begin{framed}
\begin{theorem}[Extrinsic Ceiling]

    Let \(T\) define a neutrally monotone purely extrinsic process on an
    \(N\)-player game with 2 actions, \(C\) and \(D\), such that in any given state
    \(\mathbf{a}\),
    \(\min_{j: a_j = D} \pi_j(\mathbf{a}) > \max_{i: a_i = C} \pi_i(\mathbf{a})\)
    (that is, any defector has greater payoff than any cooperator; equivalently,
    since fitness \(f_i = 1 - \varepsilon + \varepsilon \pi_i\) is increasing in payoff, any
    defector has greater fitness than any cooperator), under action-invariant
    mutation.

    Then the mean proportion of cooperators in the steady state of \(T\), denoted \(p_C\),
     satisfies \(p_C \leq \frac{1}{2}\).
    \label{prop:ceiling}
\end{theorem}
\end{framed}
\begin{proof}

    Equip \(S = \{C, D\}^N\) with the componentwise partial order, declaring \(D <
    C\). For a transition matrix \(P\), write \(P(\mathbf{a}, \cdot) \preceq_{st}
    P(\mathbf{a}', \cdot)\) to mean that the one-step distribution from
    \(\mathbf{a}\) under \(P\) is stochastically dominated by the one-step
    distribution from \(\mathbf{a}'\) under \(P\) in this order. Let
    \(T^{\mathrm{nd}}_{\mathbf{ab}}\) denote the value of \(T_{\mathbf{ab}}\)
    evaluated at \(\pi_i(\mathbf{a}) = 0\) for every \(i\); by neutral-monotonicity
    condition (1), \(T^{\mathrm{nd}}_{\mathbf{ab}}\) is equally the value of \(T_{\mathbf{ab}}\)
    under any other choice of common payoff.

    \medskip
    \noindent\textit{Step 1: \(T(\mathbf{a}, \cdot) \preceq_{st}
    T^{\mathrm{nd}}(\mathbf{a}, \cdot)\) for every \(\mathbf{a} \in S\).}

    Set \(v(\mathbf{a}) := \min_{a_j = D} \pi_j(\mathbf{a})\). In every mixed
    state \(\mathbf{a}\), the payoff-dominance hypothesis implies
    \(\pi_i(\mathbf{a}) < v(\mathbf{a})\) for every cooperator \(i\) and
    \(\pi_j(\mathbf{a}) \geq v(\mathbf{a})\) for every defector \(j\). Consider a
    \(C\)-gaining transition \(\mathbf{a} \to \mathbf{b}\), so \(b_{I(\mathbf{a,b})}
    = C\) and \(\mathcal{R}\) is the set of current cooperators. Replacing every
    \(\pi_i(\mathbf{a})\) with \(v(\mathbf{a})\) raises all role-model payoffs and
    lowers all non-role-model payoffs, so by purely-extrinsic condition (2),

    \[
        T_{\mathbf{ab}} \;\leq\; T_{\mathbf{ab}} \big|_{\pi_i(\mathbf{a}) = v(\mathbf{a})\text{ for all } i}
    \]

    with strict inequality whenever \(\mathcal{R} \neq \emptyset\). By
    neutral-monotonicity condition (1), the right-hand side equals
    \(T^{\mathrm{nd}}_{\mathbf{ab}}\). The
    analogous argument with reversed inequality applies to \(D\)-gaining
    transitions. Hence \(T(\mathbf{a}, \cdot) \preceq_{st}
    T^{\mathrm{nd}}(\mathbf{a}, \cdot)\).

    \medskip
    \noindent\textit{Step 2: \(T^{\mathrm{nd}}\) is invariant under the \(C \leftrightarrow D\) label swap.}

    Let \(\sigma : S \to S\) be the map that swaps every \(C\) and \(D\). By
    purely-extrinsic condition (3) and neutral-monotonicity condition (1), with
    common payoff \(0\), the only features of
    \((\mathbf{a}, \mathbf{b})\) that enter \(T^{\mathrm{nd}}_{\mathbf{ab}}\) are
    the mutation matrix and the set \(\mathcal{R}\). Under action-invariant
    mutation the mutation matrix is \(\sigma\)-invariant; and \(|\mathcal{R}|\) is
    \(\sigma\)-invariant because the set of players currently playing
    \(b_{I(\mathbf{a,b})}\) in \(\mathbf{a}\) has the same size as the set of
    players currently playing \(\sigma(b_{I(\mathbf{a,b})})\) in
    \(\sigma(\mathbf{a})\). Hence \(T^{\mathrm{nd}}_{\mathbf{ab}} =
    T^{\mathrm{nd}}_{\sigma(\mathbf{a})\sigma(\mathbf{b})}\) for every
    transition, so the stationary distribution of \(T^{\mathrm{nd}}\) is
    \(\sigma\)-invariant, giving \(p_C^{(T^{\mathrm{nd}})} =
    p_D^{(T^{\mathrm{nd}})} = \frac{1}{2}\).

    \medskip
    \noindent\textit{Step 3: \(T^{\mathrm{nd}}\) is stochastically monotone.}

    Neutral-monotonicity condition~(2) provides, for every \(\mathbf{a} \leq \mathbf{a}'\), a coupling
    \((\mathbf{X}, \mathbf{X}')\) on a common probability space with \(\mathbf{X}
    \sim T^{\mathrm{nd}}(\mathbf{a}, \cdot)\), \(\mathbf{X}' \sim
    T^{\mathrm{nd}}(\mathbf{a}', \cdot)\), and \(\mathbf{X} \leq \mathbf{X}'\)
    almost surely. For any upper set \(U \subseteq S\), \(\{\mathbf{X} \in U\}
    \subseteq \{\mathbf{X}' \in U\}\), so \(T^{\mathrm{nd}}(\mathbf{a}, U) \leq
    T^{\mathrm{nd}}(\mathbf{a}', U)\), giving \(T^{\mathrm{nd}}(\mathbf{a}, \cdot)
    \preceq_{st} T^{\mathrm{nd}}(\mathbf{a}', \cdot)\).

    \medskip
    Applying Lemma~\ref{lem:stationary-dominance} to \(P = T\) and \(Q =
    T^{\mathrm{nd}}\) (which satisfy (i) by Step 1 and (ii) by Step 3) yields
    \(\rho^T \preceq_{st} \rho^{T^{\mathrm{nd}}}\). Since \(n_C\) is monotone in the
    componentwise order,
    \[
        p_C^{(T)} \;=\; \frac{1}{N}\, \mathbb{E}_T[n_C]
        \;\leq\; \frac{1}{N}\, \mathbb{E}_{T^{\mathrm{nd}}}[n_C]
        \;=\; p_C^{(T^{\mathrm{nd}})} \;=\; \frac{1}{2}.
    \]
\end{proof}

Crucially, Theorem~\ref{prop:ceiling} shows that even if a player may increase
their payoff by a switch \(D \to C\) in any given state, pointwise dominance of
\(D\) will still cause the process to favour \(D\). This shows that a player's
chosen strategy does not necessarily favour the action type which maximises
their payoff in \emph{either} the one-step or the long-run strategy
distribution, instead favouring the strategy which is pointwise dominant over
its peers. Even if there is no temptation to free-ride, players will still
choose to defect based on the pointwise dominance of defectors in the state.

An example of this is the public goods game with \(r > N\). This satisfies both
pointwise dominance and a fitness increase for any player under the \(D \to C\)
swap. In such a game, there is no temptation to free ride; players will obtain a
greater payoff by cooperating regardless of the population structure. Threshold
values of \(r\) in a homogeneous public goods game are considered
in~\cite{g10010001}, where a subset of the population play a game at each time
step and players opt out of the game with a certain probability. The following
section considers the heterogeneous public goods game, providing an example of
Theorem~\ref{prop:ceiling}.

\subsubsection{Purely extrinsic population dynamics in the public goods game}

To illustrate the results of section~\ref{sec:extrinsic_population_dynamics}, we
carry out a large parameter sweep across parameters of a public goods
game~\cite{ARCHETTI20129, hauert2002volunteering, HauertDynamics2004,
hauser2019social, 2bbcfd92-fcd5-38d1-b7a0-bc6a2aadbc64}. In such a game, players
choose to contribute some amount \(\alpha_i\). The total contribution is
multiplied by some factor \(r > 0\) and split between all players. We consider
heterogeneous contributions and a homogeneous return.

The parameter sweep is done to the highest standards of
reproducibility~\cite{wilson2017, stodden2016, turingway2022, Wilson2014,
sandve2013} and all data has been archived at~\cite{data_archive} and is a
further contribution of the work. The source code and development process can be
seen at \texttt{https://github.com/hefos/The-Cooperation-Ceiling}. In order to
perform this sweep, the authors have developed an installable software library
\texttt{ludics} (written in python). The archive of \texttt{ludics} is
at~\cite{ludics_archive} and the documentation can be found
at~\cite{ludics_docs}. The full development of the software can be found
at~\cite{ludics}. The parameters considered, alongside their meanings in
reference to the public goods game, are shown in
Table~\ref{fig:parameter_sweep_table}.

\begin{table}[htbp!]
    \centering
    \caption{The set of measured parameters and boundary values. Some parameters
    are only applicable to certain processes. \(\alpha_i\) is entirely
    controlled by the value \(M\), which controls the maximum total contribution.
    The distribution of wealth remains the same for a given \(N\), however as
    \(M\) increases, so too does the size of each player's contribution
    \(\alpha_i\). The value \(A\) is the \emph{aspiration} considered in
    aspiration dynamics, in Section~\ref{sec:intrinsic_population_dynamics}. We
    consider a homogeneous return \(r\). The aspiration \(A\) depends on \(M\),
    and the selection intensity \(\varepsilon\) on both \(M\) and \(r\) in addition to \(N\)
    than on \(N\), and as such there are a large number of values per \(N\) for
    these parameters.}
    \label{fig:parameter_sweep_table}
    \begin{tabular}{|c|c|c|c|c|}
        \hline
        Parameter & Meaning &Minimum & Maximum & Number of Values per \(N\)\\
        \hline
        \(N\) & Number of players &2 & 8 & N/A\\
        \hline
        \(r\) & Return on investment & \(\frac{1}{2}\) & \(\frac{3N}{2}\) & 10\\
        \hline
        \(\alpha_i\) & Contribution of player \(i\) & N/A & N/A & N/A\\
        \hline
        \(M\) & Maximum total contribution & \(N\) & \(4N\) & 10\\
        \hline
        \(\mu\) & Mutation value & 0.001 & 0.1 & 4\\
        \hline
        \(\varepsilon\) & Selection intensity & 0 & \(\varepsilon_{\max}\), see~\eqref{eq:epsilon_bound} & 1000\\
        \hline
        \(\beta\) & Choice intensity & 0 & 2 & 10\\
        \hline
        \(A\) & Aspiration & 1 & \(M\) & 100\\
        \hline
    \end{tabular}
\end{table}

To keep every player's fitness positive, so that the Moran transition
probabilities are well defined, the selection intensity is restricted to the
interval \([0, \varepsilon_{\max}]\), where the upper bound depends on the return
\(r\) and the largest contribution \(\alpha_{\max} = \max_i \alpha_i\):
\begin{equation}
    \label{eq:epsilon_bound}
    \varepsilon_{\max} =
    \begin{cases}
        \dfrac{0.99}{1 - \alpha_{\max}\left(\frac{r}{N} - 1\right)}, & r < N,\\[2.4ex]
        \dfrac{\alpha_{\max}\left(\frac{r}{N} - 1\right) + 1}
              {\alpha_{\max}\left(\frac{r}{N} - 1\right) + 2}, & r \geq N.
    \end{cases}
\end{equation}
For each \((M, r)\) we take \(\varepsilon\) on an evenly spaced grid of ten values
in this interval.

Our results for the purely extrinsic dynamics are shown in
Figure~\ref{fig:ceiling_extrinsic}, and we explore Theorem~\ref{prop:ceiling}
numerically in the context of the 8-player public goods game; the other values
of \(N\) in our sweep are qualitatively identical. Panels (a) and (b) show this
limit as a hard ceiling on \(p_C\) which no parameter set in our data is able to
exceed. This includes parameter combinations which have high values of \(r >
N\), and so we see that despite full cooperation being the Nash equilibrium of
the one shot game, and increasing the payoff of players in the one-step
distribution, the pointwise dominance of defection over cooperation leads to
defection becoming the dominant strategy in a repeated game. This not only
leads to a worse average payoff for all players in the population in the steady
state, but also means that when \(r > N\), worsening moves occur with a higher
probability than changes of action types which improve a player's payoff. In
other words, even when a switch \(D \to C\) would always improve a player's
payoff, the player will still prefer to remain a defector under a purely
extrinsic population dynamic, due to the greater fitness of a defector in a
given state.

In Fermi imitation dynamics, the invariance under \(r\) can be shown directly.
In the public goods game the multiplied contribution
\(\frac{r}{N}\sum_{a_k = C}\alpha_k\) is shared equally by every player, so for
any two players \(i\) and \(j\) in a state \(\mathbf{a}\) this common term
cancels in their payoff difference:
\begin{equation}
    \pi_i(\mathbf{a}) - \pi_j(\mathbf{a})
    = \alpha_j \mathbf{1}_{a_j = C} - \alpha_i \mathbf{1}_{a_i = C},
    \label{eq:fermi_invariant}
\end{equation}
which does not depend on \(r\). As Fermi imitation dynamics depends on payoffs
only through such pairwise differences, its transition matrix and steady state,
and hence \(p_C\), are invariant to \(r\); this is confirmed in
Figure~\ref{fig:ceiling_extrinsic}(d). An increase in the choice intensity
\(\beta\) has a negative effect upon \(p_C\), as the value \(\Delta(\pi)\) will
always be negative for a \(C \to D\) transition, and positive for a
\(D \to C\) transition. As the
process with \(\beta=0\) gives us neutral drift, we see our result in practice,
as moving towards a less rational process increases \(p_C\) in all empirical
cases, up to the neutral-drift baseline. The notion of ``rationality'' in a
purely extrinsic process is not directly linked to increasing a player's payoff.
\(p_C\) in the Moran process increases with the value of \(r\), however once
again we see that it has a hard limit at \(p_C = \frac{1}{2}\), as a given
defector is always more likely to be chosen for reproduction than a given
cooperator.

One interesting point is the invariance of the Moran process under the size of
each player's contribution. The proportion of the state's total fitness given by
that of defectors is:
\begin{equation}
    \frac{qr\sum_{j=1}^N \alpha_j}{Nr\sum_{j=1}^N \alpha_j - \sum_{\mathbf{a}_l = C}\alpha_l} =
    \frac{qr\sum_{j=1}^N \alpha_j}{(Nr)\sum_{j=1}^N \alpha_j - \sum_{j=1}^N \alpha_j}=
    \frac{qr\sum_{j=1}^N \alpha_j}{(Nr-1)\sum_{j=1}^N \alpha_j} = \frac{qr}{Nr-1}
    \label{eq:Moran_invariant}
\end{equation}
The analogous result for choosing a defector is given by \(1 -
\frac{qr}{Nr-1}\), which is once again invariant under \(\alpha_i\). Therefore,
as Fermi imitation dynamics is invariant under the return parameter \(r\), the
Moran process is invariant under the contribution parameter \(\alpha_i\). We can
observe this effect in Figure~\ref{fig:ceiling_extrinsic}(c).

\begin{figure}[!htbp]
    \centering
    \begin{mdframed}

    \includegraphics[width=\linewidth]{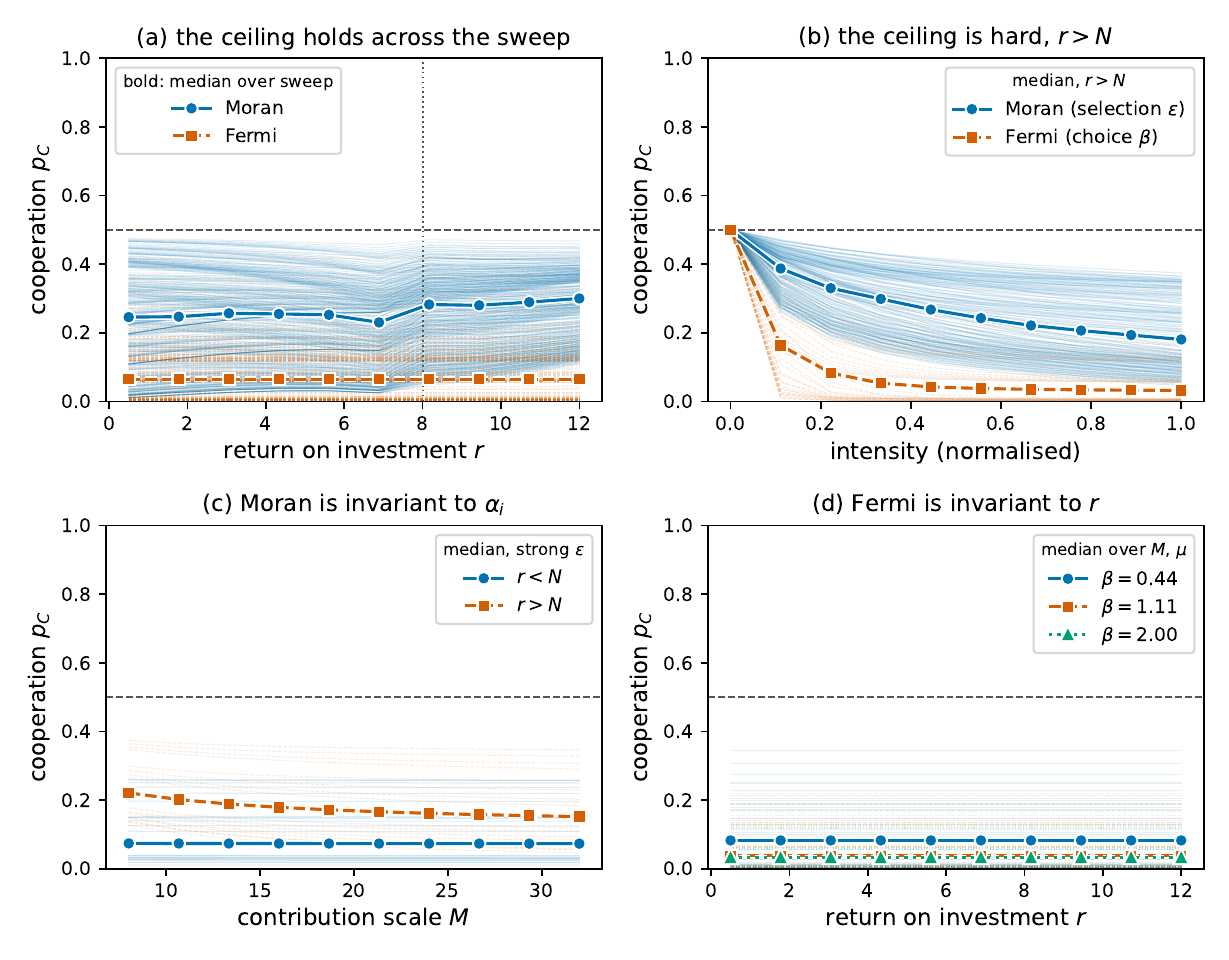}

    \end{mdframed}

    \caption{\textbf{The cooperation ceiling under purely extrinsic dynamics.}
    Throughout, \(N = 8\), a linear contribution profile, and mutation
    \(\mu = 0.05\); the dashed line marks the neutral-drift baseline
    \(p_C = \frac{1}{2}\) and, in (a), the dotted line marks \(r = N\). Faint
    curves show \(p_C\) trajectories across the sweep, coincident trajectories
    drawn once; in (a) they join parameter sets at equal selection- or
    choice-intensity rank across \(r\), since the selection-intensity grid is
    rescaled with \(r\) (Equation~\eqref{eq:epsilon_bound}). Bold curves are the
    median over all parameter sets. \textbf{(a)} Cooperation \(p_C\) against the
    return \(r\): the Moran process rises with \(r\), with a dip just below
    \(r = N\) where the selection-intensity bound is widest, but never exceeds
    \(\frac{1}{2}\); Fermi imitation dynamics stays well below it. \textbf{(b)}
    With \(r > N\), \(p_C\)
    against the normalised (dividing by the maximum value) selection intensity 
    \(\varepsilon\) (Moran) and choice
    intensity \(\beta\) (Fermi); neither exceeds \(\frac{1}{2}\), reaching it only
    as the intensity vanishes. \textbf{(c)} The Moran process is invariant to the
    contribution scale \(M\), as derived in Equation~\eqref{eq:Moran_invariant}.
    \textbf{(d)} Fermi imitation dynamics is invariant to the return \(r\), as
    derived in Equation~\eqref{eq:fermi_invariant}.}
    \label{fig:ceiling_extrinsic}
\end{figure}

\section{Intrinsic population dynamics}
\label{sec:intrinsic_population_dynamics}

We have seen how purely extrinsic population dynamics only compare a player's
payoffs to the payoff in the current state. We now see a classification of
population dynamic where players instead compare their payoff to some attribute
or payoff other than that of other players. We define a \emph{purely intrinsic
population dynamic} as follows:
\begin{definition}[Purely intrinsic population dynamic]
    A population dynamic defined by a transition matrix \(T\) 
    is said to be \emph{purely intrinsic} if it satisfies the following conditions:
    \begin{enumerate}
        \item \(T_{\mathbf{ab}}\) is a function of the payoffs of no player
        other than \(I(\mathbf{a,b})\) for \(\mathbf{a} \neq \mathbf{b}\).
        \item \(T_{\mathbf{ab}}\) is decreasing in the payoff
        \(\pi_{I(\mathbf{a,b})}(\mathbf{a})\) for \(\mathbf{a} \neq \mathbf{b}\).
        \item \(T_{\mathbf{ab}}\) contains no term dependent on
        \(a_{I(\mathbf{a,b})}\) or \(b_{I(\mathbf{a,b})}\), except in the payoff
        function \(\pi_{I(\mathbf{a,b})}\) and mutation parameters
        \(\mu_{ij}\).
    \end{enumerate}
\end{definition}

We now consider two purely intrinsic population dynamics: aspiration dynamics,
which compares a player's payoff to their individual baseline aspiration, and
introspection dynamics, which compares a player's payoff to their counterfactual
payoff when they change action type.

\subsection{Aspiration dynamics}

Under aspiration dynamics, players choose to change action based on how far away
their current payoff is from a target payoff, as shown in
Figure~\ref{fig:diagrams_of_dynamics}B. It
follows the algorithm:

\begin{enumerate}
       \item A player \(i\) is chosen with probability \(\frac{1}{N}\) to
       reconsider their strategy
       \item This player changes strategy with a probability
       \(\phi(\pi_i(\mathbf{a}) - A_i)\)
\end{enumerate}

Here, \(A_i\) is player \(i\)'s \textit{aspiration}, the payoff which they
desire to make from a given game. A player will accept a new strategy with a
higher than random probability if their current payoff is lower than their
aspired payoff. This dynamic is traditionally paired with two action
games~\cite{Du2015AspirationDynamics,du2014aspiration,Zhou2021AspirationDynamics},
though it can be defined for games with 3 or more
actions~\cite{Aspiration3Types}. Aspiration dynamics has also often been
combined with imitation
dynamics~\cite{PhysRevE.102.032120,PhysRevE.94.012124,QUAN2020109634,WANG2023128134},
either as a heterogeneous attribute of the population, or to form new population
dynamics. We consider the original form of aspiration dynamics defined
in~\cite{du2014aspiration}, with a homogeneous aspiration value \(A\). The
transition matrix is given by:

\begin{equation}
T_{\mathbf{ab}} =
\begin{cases}
\dfrac{1}{N} \, \phi(\pi_{I(\mathbf{a}, \mathbf{b})}(\mathbf{a}) - A_{I(\mathbf{a}, \mathbf{b})})\left(1 - \sum_{j=1}^k \mu_{I(\mathbf{a}, \mathbf{b}), j}\right)+ \frac{\mu_{I(\mathbf{a}, \mathbf{b}), \mathbf{b}_{I(\mathbf{a}, \mathbf{b})}}}{N}
    & \text{if } \mathbf{b} \in \mathrm{Neb}(\mathbf{a}),\\[1.2em]
0
    & \text{if } \mathbf{b} \notin \mathrm{Neb}(\mathbf{a}) \text{ and } \mathbf{a} \neq \mathbf{b},\\[0.8em]
1 - \sum_{\mathbf{b} \in S \setminus \{\mathbf{a}\}} T_{\mathbf{ab}} & \text{if } \mathbf{a} = \mathbf{b}.
\end{cases}
\label{eqn:Aspiration_Probability_Function}
\end{equation}

This population dynamic spends a large amount of time in states where most
players achieve their aspired payoffs. In this case, players will have a lower
probability of switching strategy, compared to a state where players are
underperforming in comparison to their target \(A\). We now show that Aspiration
dynamics is a purely intrinsic population dynamic:

\begin{theorem}[Aspiration dynamics is a purely intrinsic population dynamic]
    Aspiration dynamics, as defined by
    Equation~\eqref{eqn:Aspiration_Probability_Function}, is a purely intrinsic
    population dynamic.
\end{theorem}
\begin{proof}
    It is clear from the definition of \(T_{\mathbf{ab}}\) that aspiration
    dynamics satisfies conditions (1) and (3). It remains to show that
    aspiration dynamics is decreasing in the payoff
    \(\pi_{I(\mathbf{a,b})}(\mathbf{a})\). We have that \(\phi(x)\) is
    decreasing, and \(\pi_{I(\mathbf{a,b})}(\mathbf{a})\) appears only in the
    value \(\pi_{I(\mathbf{a,b})}(\mathbf{a}) - A_{I(\mathbf{a,b})}\).
    Therefore, setting \(x = \pi_{I(\mathbf{a,b})}(\mathbf{a}) -
    A_{I(\mathbf{a,b})}\), we get that \(\phi\) is decreasing in
    \(\pi_{I(\mathbf{a,b})}(\mathbf{a})\), and therefore so is
    \(T_{\mathbf{ab}}\).
\end{proof}

\subsection{Introspection dynamics}

Unlike the previous population dynamics we have seen, introspection
dynamics~\cite{Couto2022, Couto2023} allows players to consider the payoff they
could have in other states. Under this population dynamic, players compare their
current payoff with their counterfactual payoff upon changing action type. The
process, as shown in Figure~\ref{fig:diagrams_of_dynamics}D, is given by the
algorithm:

\begin{enumerate}
    \item A player \(i\) is chosen with probability \(\frac{1}{N}\) to
    reconsider their strategy
    \item Player \(i\) randomly chooses another possible strategy, \(a_{il}\),
    from the set of all possible strategies. 
    \item The chosen player replaces their strategy with probability \(\phi
    (\Delta \pi_{I(\mathbf{a}, \mathbf{b})})\), where \(\Delta
    \pi_{I(\mathbf{a}, \mathbf{b})} = \pi_{I(\mathbf{a},
    \mathbf{b})}(\mathbf{a})  - \pi_{I(\mathbf{a}, \mathbf{b})}(\mathbf{b})\) is
    the difference between the payoff of the player's current strategy and the
    payoff of the new strategy.
\end{enumerate}

Introspection dynamics is well suited to a heterogeneous population, as players
are able to consider whether a strategy will work for them, rather than for
another player. As players check the payoff of a strategy in the next state,
they never favour worsening moves in the one-step case, as the value \(\phi
(\Delta \pi_{I(\mathbf{a}, \mathbf{b})})\) is greater than \(\frac{1}{2}\) if
and only if the change in a player's strategy increases their payoff. The
transition matrix is given by:

\begin{equation}
T_{\mathbf{ab}} =
\begin{cases}
\dfrac{1}{N(k - 1)} \, \phi (\Delta (\pi_{I(\mathbf{a}, \mathbf{b})}))\left(1 - \sum_{j=1}^k \mu_{I(\mathbf{a}, \mathbf{b}), j}\right)+ \frac{\mu_{I(\mathbf{a}, \mathbf{b}), \mathbf{b}_{I(\mathbf{a}, \mathbf{b})}}}{N}
    & \text{if } \mathbf{b} \in \mathrm{Neb}(\mathbf{a}),\\[1.2em]
0
    & \text{if } \mathbf{b} \notin \mathrm{Neb}(\mathbf{a}) \text{ and } \mathbf{a} \neq \mathbf{b},\\[0.8em]
1 - \sum_{\mathbf{b} \in S \setminus \{\mathbf{a}\}} T_{\mathbf{ab}} & \text{if } \mathbf{a} = \mathbf{b}.
\end{cases}
\label{eqn:Introspection_Probability_Function}
\end{equation}

Introspection dynamics is aligned with classic static game theory, where a well
informed and rational actor maximises their own payoff. We can contrast this
with purely extrinsic population dynamics, which are the traditional method of
modelling populations in evolutionary games, where actors are not typically
well-informed and choose their next strategy according to fitnesses in the
current population, without knowledge of the impact of their transition.

\begin{theorem}[Introspection dynamics is a purely intrinsic population dynamic]
    Introspection dynamics, as defined by
    Equation~\eqref{eqn:Introspection_Probability_Function}, is a purely
    intrinsic population dynamic.
\end{theorem}
\begin{proof}
    It is once again clear that introspection dynamics satisfies conditions (1) and (3).
    For condition (2), consider that \(\phi(x)\) is decreasing. Setting 
    \(x = \pi_{I(\mathbf{a,b})}(\mathbf{a})- \pi_{I(\mathbf{a,b})}(\mathbf{b})\), we have our result.
\end{proof}

Unlike purely extrinsic population dynamics, introspection dynamics does not
necessarily cancel out certain parameters of games, unless said parameter does
not affect the difference in a player's payoff when playing different action
types. For example, if we consider the public goods game, we can see that for
introspection dynamics,
\begin{equation}
    \Delta(\pi_{I(\mathbf{a}, \mathbf{b})}) = \alpha_i\!\left(\frac{r}{N} - 1\right),
    \label{eq:introspection_threshold}
\end{equation}
which considers both \(\alpha_i\) and \(r\), and changes sign at \(r = N\).

\subsection{Intrinsic population dynamics in a public goods game}

\begin{figure}[!htbp]
    \centering
    \begin{mdframed}
    \includegraphics[width=\linewidth]{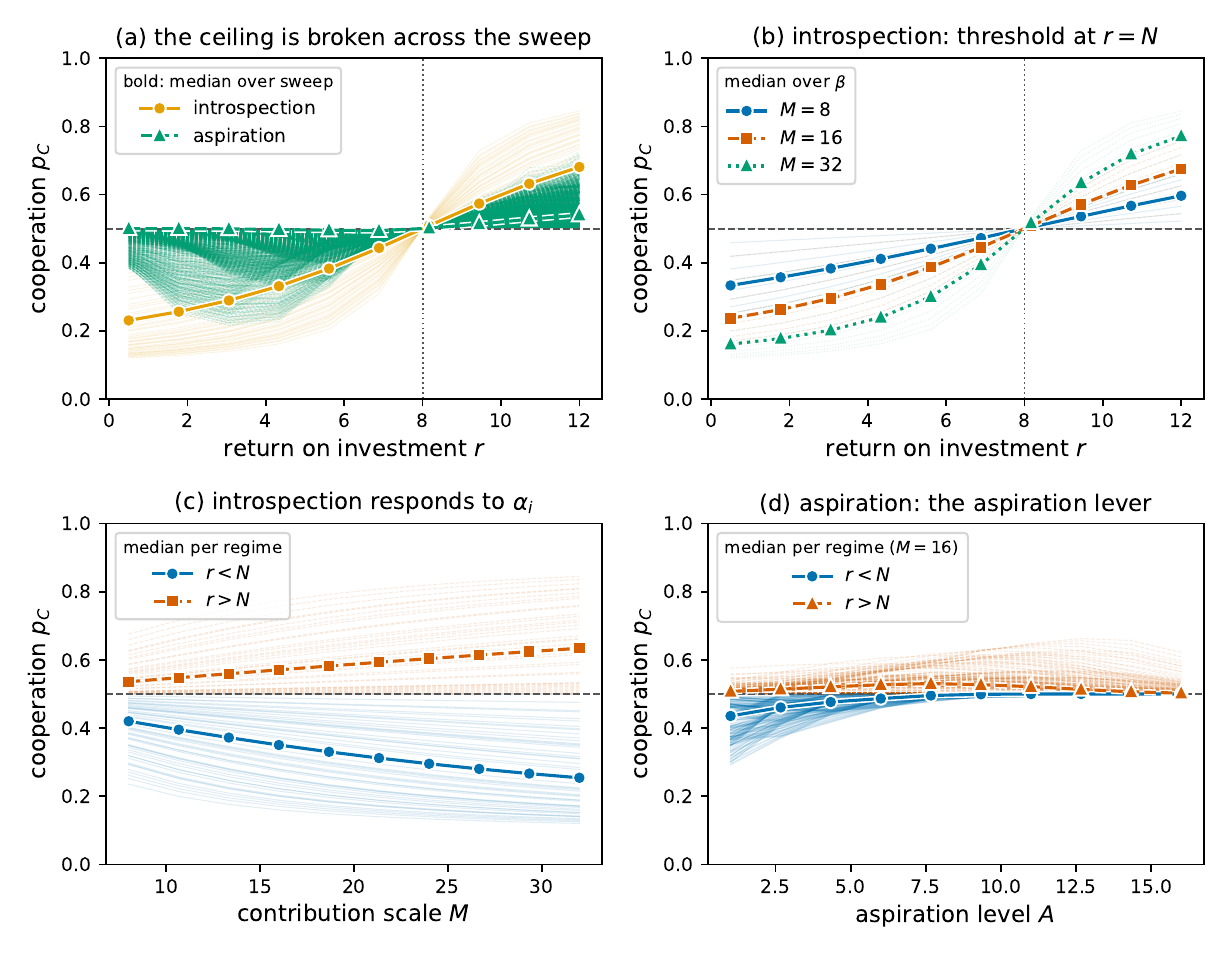}

    \end{mdframed}

    \caption{\textbf{Intrinsic dynamics exceed the cooperation ceiling.} Same
    population as Figure~\ref{fig:ceiling_extrinsic}; faint curves show the
    \(p_C\) trajectory of every parameter set in the sweep (coincident
    trajectories are drawn once), and the bold curves their median taken over all
    parameter sets, with the dashed line at \(p_C = \frac{1}{2}\) and the dotted
    line at \(r = N\).
    \textbf{(a)} Cooperation \(p_C\) against the return \(r\) for introspection and
    aspiration dynamics; both exceed \(\frac{1}{2}\) once \(r > N\). \textbf{(b)}
    Introspection \(p_C\) against \(r\) at several contribution scales \(M\), every
    curve crossing \(\frac{1}{2}\) exactly at \(r = N\): the threshold is set by the
    sign of the counterfactual payoff difference
    \(\Delta(\pi_{I(\mathbf{a},\mathbf{b})})\) in
    Equation~\eqref{eq:introspection_threshold}. \textbf{(c)} Introspection \(p_C\)
    against \(M\): it increases with the contribution scale when \(r > N\) and
    decreases when \(r < N\). \textbf{(d)} Aspiration \(p_C\) against the aspiration
    level \(A\) at \(M = 16\).}
    \label{fig:intrinsic_PGG}
\end{figure}

In Figure~\ref{fig:intrinsic_PGG}, we see that purely intrinsic population
dynamics are able to exceed the ceiling on \(p_C\) which is imposed on purely
extrinsic population dynamics. In the case of aspiration dynamics, \(p_C\) does
not reach excessively low values as in the other processes we have discussed.
This is because for all players to have a payoff \(\pi_i > 0\), at least one
player must be cooperating at all times. If no players cooperate, then no player
will reach their target payoff threshold, and players will switch strategy with
a high probability. The ceiling is not exceeded for \(r < N\), as shown in
Figure~\ref{fig:intrinsic_PGG}(a) and noted for a linear game
in~\cite{du2014aspiration}, but is exceeded at the threshold \(r > N\). In this
case, \(\pi_i(\mathbf{a}) - A\) is greater when cooperating than defecting in a
given population \(\mathbf{a}_{-i}\) for all players, and thus the transition \(D \to
C\) will always occur with a higher probability than the reverse. For
intermediate values of \(r\), aspiration dynamics can achieve its lowest
\(p_C\), as with a low value of \(A\), players are able to achieve their
aspiration by free-riding in a state with fewer cooperators, and so there is no
incentive to begin cooperating in such a state.
Figure~\ref{fig:intrinsic_PGG}(d) shows how \(p_C\) depends on the aspiration
level \(A\). For \(r > N\), \(p_C\) is non-monotone in \(A\): it rises to a peak
and then falls back towards \(\frac{1}{2}\). Aspiration only discriminates
between the two actions when \(A\) lies between the payoffs of a defector and a
cooperator, or else both transitions \(D \to C\) or \(C \to D\) will be favoured
or surpressed. A low \(A\) leaves defectors content and unwilling to switch,
whereas an \(A\) above even the cooperator's payoff leaves every player
dissatisfied, so switching becomes action-independent and the advantage of
cooperation is washed out. The aspiration level is the biggest indicator of
cooperation in a linear public goods game where a player's fitness is given by
their expected return in a 2-player game. It has been shown previously that for
a high aspiration level, cooperation emerges at a level near the neutral-drift
baseline, but for a low aspiration level defectors dominate the
population~\cite{du2014aspiration}.

As the heterogeneous public goods game is an additive game~\cite{Couto2023}, we
can show that many of the analytical results of~\cite{foster_knight_2026} about
the structure of introspection dynamics in a public goods game hold empirically.
Each player \(i\) cooperates with probability \(p_i = \phi_i(\alpha_i(1 -
\frac{r}{N}))(1 - 2\mu) + \mu\), and \(p_C\) is given as a product
measure~\cite{Couto2023}. Thus, we have that as the value
\(\Delta(\pi_{I(\mathbf{a,b})}) = \alpha_i(\frac{r}{N} - 1)\) is positive for
\(r > N\) and negative for \(r < N\), there is a sharp turning point at \(r =
N\) where introspection dynamics begins to favour cooperation, as shown in
Figure~\ref{fig:intrinsic_PGG}(b). This gives an exact parameter value at
which the ceiling of \(p_C\) is overcome.

Figure~\ref{fig:intrinsic_PGG}(c) shows the result that \(\frac{\partial
p_C}{\partial \alpha_i} > 0\) if and only if \(r > N\), as the value
\(\alpha_i(\frac{r}{N} - 1)\) becomes greater in this case, but becomes more
negative if \(r < N\). A greater contribution will increase a player's payoff if
multiplied by a large \(r\), but for a small \(r\) a larger contribution results
in a greater loss to the player.

The same conditions also apply to the effect of \(\beta\). This is because
\(\phi_i(\Delta(\pi_{I(\mathbf{a,b})}))\) is decreasing. Therefore, as \(\beta\)
scales \(\Delta(\pi_{I(\mathbf{a,b})})\), when \(\Delta(\pi_{I(\mathbf{a,b})})\)
is negative, \(\frac{\partial p_C}{\partial \alpha_i} > 0\), and the inequality
is reversed if the opposite is true. We acquire \(\frac{\partial p_C}{\partial
\alpha_i} < 0\) at exactly the value \(r > N\), and thus we have our result.
These results are proven in~\cite{foster_knight_2026}, but have been shown
empirically for our large data sweep in Figure~\ref{fig:intrinsic_PGG}.

Pooling the entire parameter sweep makes the contrast between intrinsic and
extrinsic dynamics explicit. Figure~\ref{fig:ceiling_overview} collects every
parameter combination with \(\mu > 0\), across all \(N\), contributions,
returns, intensities and aspirations. No purely extrinsic parameter set anywhere
in the sweep exceeds the neutral-drift baseline \(p_C = \frac{1}{2}\), and the
largest \(p_C\) attained is pinned at exactly \(\frac{1}{2}\) for every \(N\).
The intrinsic dynamics, by contrast, place a substantial fraction of their mass
above the ceiling, and this fraction switches on precisely at \(r = N\)
(Table~\ref{tab:ceiling_crossing}).

\begin{figure}[!htbp]
    \centering
    \begin{mdframed}
    \includegraphics[width=\linewidth]{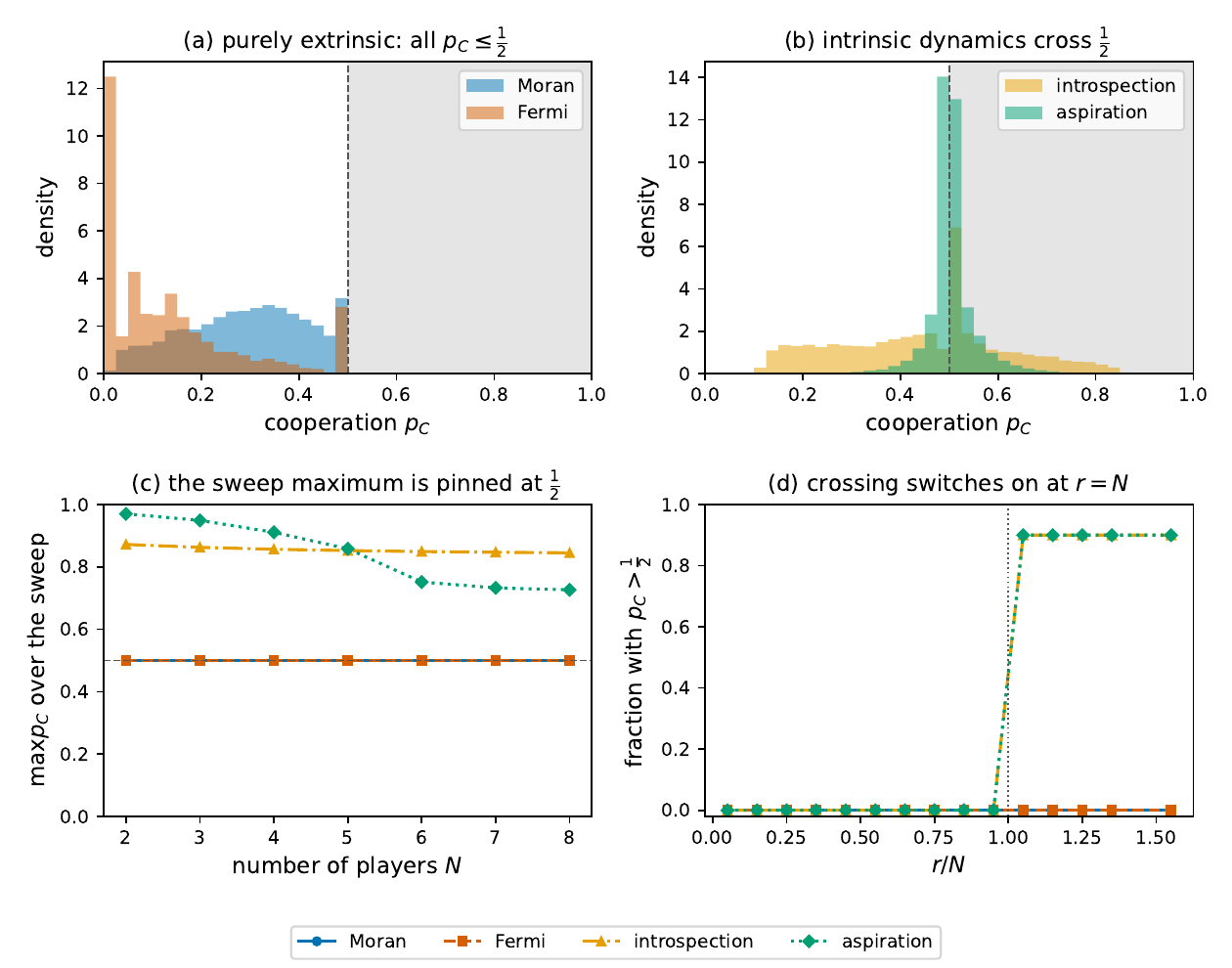}
    \end{mdframed}
    \caption{\textbf{The cooperation ceiling across the entire sweep.} Each panel
    pools every parameter combination with \(\mu > 0\). \textbf{(a)} The
    distribution of \(p_C\) for the purely extrinsic dynamics lies entirely at or
    below the neutral-drift baseline \(p_C = \frac{1}{2}\) (shaded region).
    \textbf{(b)} The intrinsic dynamics place mass above \(\frac{1}{2}\).
    \textbf{(c)} The maximum \(p_C\) attained anywhere in the sweep, as a function
    of \(N\): for the extrinsic dynamics it is pinned at exactly \(\frac{1}{2}\) at
    every \(N\). \textbf{(d)} The fraction of parameter sets with
    \(p_C > \frac{1}{2}\) as a function of \(r / N\): it is zero for the extrinsic
    dynamics at every return, and for the intrinsic dynamics switches on precisely
    at \(r = N\).}
    \label{fig:ceiling_overview}
\end{figure}

\begin{table}[!htbp]
    \centering
    \caption{Percentage of parameter sets whose steady state crosses the
    cooperation ceiling, \(p_C > \frac{1}{2}\), by dynamic and return regime,
    over the full sweep with \(\mu > 0\).}
\begin{tabular}{lrrrr}
\hline
dynamic & parameter sets & all & \(r < N\) & \(r > N\) \\
\hline
Moran & 28,000 & 0.0\% & 0.0\% & 0.0\% \\
Fermi & 28,000 & 0.0\% & 0.0\% & 0.0\% \\
introspection & 28,000 & 36.0\% & 0.0\% & 90.0\% \\
aspiration & 280,000 & 36.0\% & 0.0\% & 90.0\% \\
\hline
\end{tabular}

    \label{tab:ceiling_crossing}
\end{table}

\section{Robustness to population size}
\label{sec:robustness_to_population_size}

The exact steady-state analysis and the parameter sweep so far have used small
populations, for which the transition matrix can be formed explicitly. Since the
size of the state space grows as \(2^N\), this is infeasible for large \(N\). To
confirm that the cooperation ceiling and its breaking are not artefacts of small
populations, we estimate \(p_C\) for populations as large as \(N = 100\) by
simulating the Markov chain directly with~\cite{ludics_archive} and averaging
the cooperator fraction over an ergodic run with mutation \(\mu = 0.05\).

Figure~\ref{fig:large_n} shows that the conclusions of the paper survive at
\(N = 100\). Panel~\ref{fig:large_n}A shows that the purely extrinsic dynamics
remain at or below the neutral-drift baseline \(p_C = \frac{1}{2}\) for all
returns; introspection dynamics rises clearly through \(\frac{1}{2}\) at
\(r = N\), while aspiration dynamics stays close to it. Panel~\ref{fig:large_n}B
shows the cooperation probability at a high return (\(r / N = 1.3\)) for each
population size: the extrinsic dynamics stay well below \(\frac{1}{2}\) across the
whole range \(10 \leq N \leq 100\), introspection lies clearly above it, and
aspiration sits just above \(\frac{1}{2}\) with a margin that narrows as \(N\)
grows. Panel~\ref{fig:large_n}C takes introspection at a different choice
intensity \(\beta\) for each of \(N \in \{10, 50, 100, 200\}\): the curves differ
in steepness, yet every one of them crosses \(\frac{1}{2}\) at exactly
\(r = N\), so the threshold is robust to both the population size and the choice
intensity. Panel~\ref{fig:large_n}D validates the
simulator against the exact steady state for \(2 \leq N \leq 8\): the simulated
estimates bracket the exact value for all four dynamics. The extrinsic ceiling is therefore unbroken at every
population size, and an intrinsic dynamic continues to exceed it for large \(N\).

\begin{figure}[!htbp]
    \centering
    \begin{mdframed}
    \includegraphics[width=\linewidth]{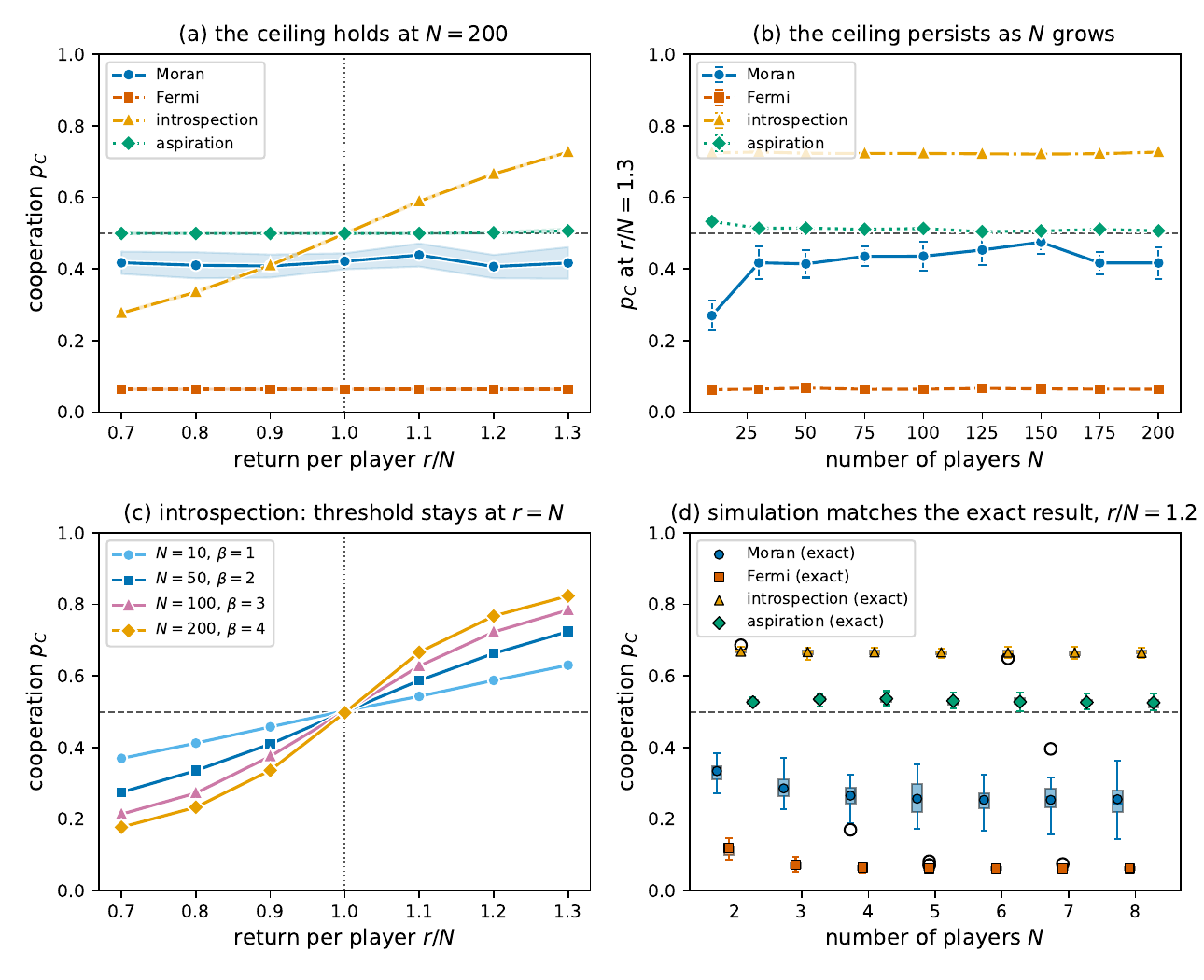}
    \end{mdframed}
    \caption{\textbf{The cooperation ceiling persists for large populations.}
    Cooperation \(p_C\) is estimated by simulating the Markov chain and averaging
    the cooperator fraction over the run. Throughout, mutation \(\mu = 0.05\), a
    linear contribution profile of scale \(M = 2N\), selection intensity
    \(\varepsilon = 0.5\) for the Moran process, choice intensity \(\beta = 2\) for
    Fermi, introspection and aspiration dynamics, and aspiration level
    \(A = 0.7\,M\); bands, error bars and boxes show the variation over
    independent seeds.
    \textbf{A,} \(p_C\) against \(r / N\) at \(N = 100\): the
    extrinsic dynamics (Moran, Fermi) stay at or below \(\frac{1}{2}\);
    introspection crosses it at \(r = N\), while aspiration only marginally
    exceeds it. \textbf{B,} \(p_C\) at a high return (\(r / N = 1.3\)), against the
    population size \(N\): below \(\frac{1}{2}\) for the extrinsic dynamics at
    every \(N\), clearly above it for introspection, and just above it for
    aspiration. Each point carries a 95\% error bar over seeds; for Fermi,
    introspection and aspiration the variation is so small that the bars are
    smaller than the markers, while the Moran process shows visible sampling
    noise that shrinks as \(N\) grows. \textbf{C,} Introspection \(p_C\) against \(r / N\), each \(N\) at
    a different choice intensity \(\beta\): the curves differ in shape but all
    cross \(\frac{1}{2}\) at \(r = N\).
    \textbf{D,} Validation against the exact steady state for \(2 \leq N \leq 8\)
    at \(r / N = 1.2\): the boxes show the simulated estimates over seeds and the
    markers the exact values, for all four dynamics.}
    \label{fig:large_n}
\end{figure}

\section{Conclusion}

We have introduced and analysed a framework for heterogeneous evolutionary games
under four distinct population dynamics in two categories, as shown in
Figure~\ref{fig:diagrams_of_dynamics}. We have defined the categorisation of a
\emph{purely extrinsic population dynamic}, and given rigorous formulations of
two widely used examples: the Moran process and Fermi imitation dynamics. We
have shown that dynamics under this classification have a hard ceiling on the
probability of cooperation \(p_C\) in social dilemmas where a defector
\emph{always} outperforms a cooperator in a given state. In such a game, purely
extrinsic population dynamics are unable to exceed the neutral-drift baseline.
We then give formulations of two purely intrinsic population dynamics:
aspiration dynamics in which players change strategy based on their performance
in relation to a target payoff, and introspection dynamics in which players
compare their current payoff with their counterfactual payoff. We have also
shown empirical results for the public goods game with heterogeneous
contributions which confirm the extrinsic ceiling, as well as the ability for
purely intrinsic population dynamics to exceed it.

Further work may consider populations in which the choice between intrinsic and
extrinsic update rules is itself a heterogeneous attribute of the players, in
the spirit of the combination of Moran and Fermi processes in~\cite{LIU2015242}.
Figure~\ref{fig:conclusion-fig}(a) illustrates this for a population that mixes
Moran and introspection players: once the return exceeds the threshold \(r = N\)
the population can exceed the neutral-drift baseline, though the number of
intrinsic players required depends on the relative strength of the two steps. A
stronger extrinsic selection intensity \(\varepsilon\), or a weaker choice
intensity \(\beta\), demands more intrinsic players before \(p_C\) crosses the
baseline. For \(r < N\) the same intrinsic players instead suppress cooperation.
The heterogeneity considered in this paper is also limited to the players'
contributions, whereas the introspection threshold acts at the level of each
player's own return. Figure~\ref{fig:conclusion-fig}(b) shows that under a
heterogeneous return \(r_i\), cooperation crosses the ceiling once enough
players sit above the \(r = N\) threshold. The interplay between such
heterogeneous returns, contributions and update rules is a natural direction for
further study. Finally, we only consider the ceiling of \(p_C\) in two-action
games for purely extrinsic population dynamics. It is entirely possible that
there is a formulation of a \(k\)-action game such that a ceiling is posed on
the abundance of a certain action. This may be a focus of future work, extending
the limitation of purely extrinsic population dynamics to a larger set of games.

\begin{figure}[!htbp]
    \centering
    \begin{mdframed}
    \includegraphics[width=\linewidth]{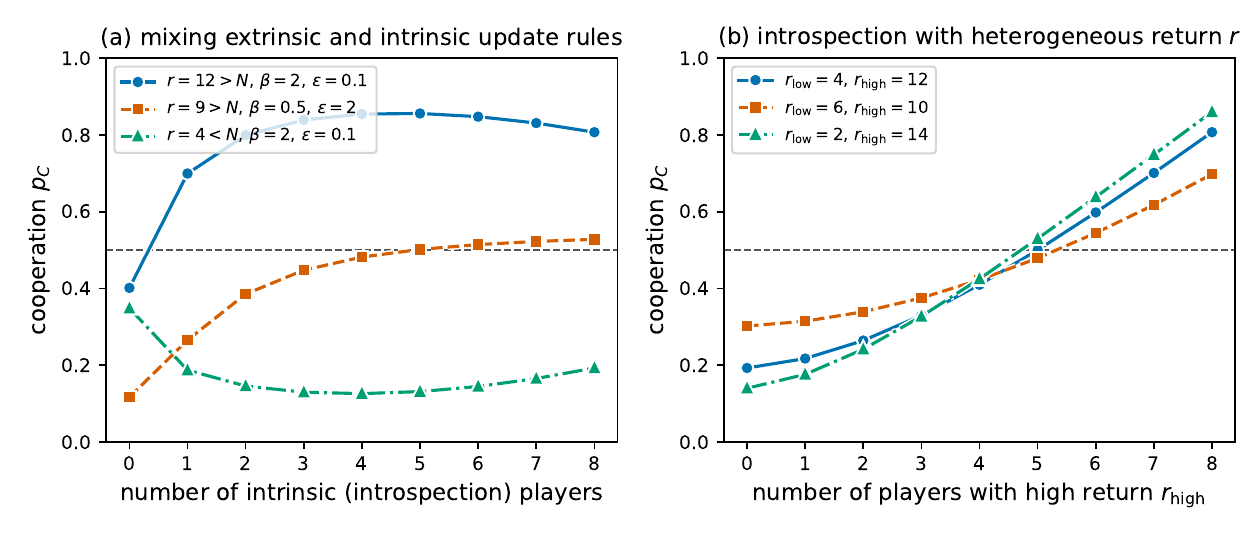}
    \end{mdframed}
    \caption{\textbf{Heterogeneous update rules and heterogeneous returns.} Both
    panels are computed exactly from the steady state of the corresponding
    transition matrix, with \(N = 8\), a linear contribution profile of scale
    \(M = 16\), and mutation \(\mu = 0.05\); the dashed line marks the
    neutral-drift baseline \(p_C = \frac{1}{2}\). \textbf{(a)} A population that
    mixes Moran (extrinsic) and introspection (intrinsic) players, against the
    number of introspection players: with \(r > N\) the population can exceed the
    baseline, but the number of intrinsic players required grows as the extrinsic
    selection intensity \(\varepsilon\) increases or the choice intensity
    \(\beta\) falls; with \(r < N\) the intrinsic players instead suppress
    cooperation. \textbf{(b)} Introspection dynamics with a heterogeneous return,
    against the number of players receiving the high return \(r_{\mathrm{high}}\);
    cooperation crosses the ceiling once enough players sit above the threshold
    \(r = N\).}
    \label{fig:conclusion-fig}
\end{figure}

\section{Acknowledgements}
Harry Foster's research was supported by EPSRC grant EP/Z535126/1.

\bibliographystyle{plain}
\bibliography{bibliography.bib}
\end{document}